%% file: Driver_NN_Pressure.tex
\documentclass[11pt,reqno]{amsproc}
\allowdisplaybreaks
\usepackage{fullpage}
\usepackage[semicolon,square,authoryear]{natbib}
\numberwithin{equation}{section}
\usepackage{cite}
\usepackage{enumerate}
\usepackage{caption}
\usepackage{color}
\usepackage[section]{placeins}
\usepackage{graphicx}
\graphicspath{{Figures/}}
\usepackage{psfrag,epsfig}
\usepackage{epstopdf}
\usepackage{subfigure}
\usepackage{amsmath}

\usepackage[debug=false, colorlinks=true, pdfstartview=FitV, linkcolor=blue, citecolor=blue, urlcolor=blue]{hyperref}

\usepackage{morefloats}

\usepackage{multirow}

\usepackage{booktabs}

\usepackage{subfigure}
\usepackage{placeins}

\newlength{\drop}
\definecolor{amethyst}{rgb}{0.6, 0.4, 0.8}
\definecolor{burgundy}{rgb}{0.5, 0.0, 0.13}
\newtheorem{theorem}{Theorem}[section]
\newtheorem{lemma}[theorem]{Lemma}

\title{\textbf{A scalable variational inequality 
approach for flow through porous media 
models with pressure-dependent viscosity}}

\author{\textbf{N.~K.~Mapakshi}, \textbf{J.~Chang} and \textbf{K.~B.~Nakshatrala} \\
{\small Department of Civil and Environmental Engineering, 
  University of Houston. \\
  \textbf{Correspondence to:}~\textsf{knakshatrala@uh.edu}}}

\date{\today}
\begin{document}

\begin{titlepage}
  \drop=0.1\textheight
  \centering
  \vspace*{\baselineskip}
  \rule{\textwidth}{1.6pt}\vspace*{-\baselineskip}\vspace*{2pt}
  \rule{\textwidth}{0.4pt}\\[\baselineskip]
  {\LARGE \textbf{\color{burgundy}
  A scalable variational inequality approach for  \\[0.2\baselineskip] 
  flow through porous media models with \\[0.2\baselineskip]
  pressure-dependent viscosity}}\\[0.3\baselineskip]
    \rule{\textwidth}{0.4pt}\vspace*{-\baselineskip}\vspace{3.2pt}
    \rule{\textwidth}{1.6pt}\\[0.5\baselineskip]
    \scshape
    An e-print of the paper will be made available on arXiv. \par
    \vspace*{0.5\baselineskip}
    Authored by \\[0.3\baselineskip]
    {\Large N.~K.~Mapakshi \par}
    {\itshape Graduate Student, University of Houston. \par}
    \vspace*{0.2\baselineskip}
    {\Large J.~Chang \par}
    {\itshape Postdoctoral Researcher, Rice University. \par}
    \vspace*{0.2\baselineskip}
    {\Large K.~B.~Nakshatrala\par}
    {\itshape Department of Civil \& Environmental Engineering \\
    University of Houston, Houston, Texas 77204--4003. \\ 
    \textbf{phone:} +1-713-743-4418, \textbf{e-mail:} knakshatrala@uh.edu \\
    \textbf{website:} http://www.cive.uh.edu/faculty/nakshatrala\par}    
    \vspace*{0.2\baselineskip}
   \begin{figure}[h]
   	\centering
   	\subfigure[RT0 formulation]{\includegraphics[scale=0.29]{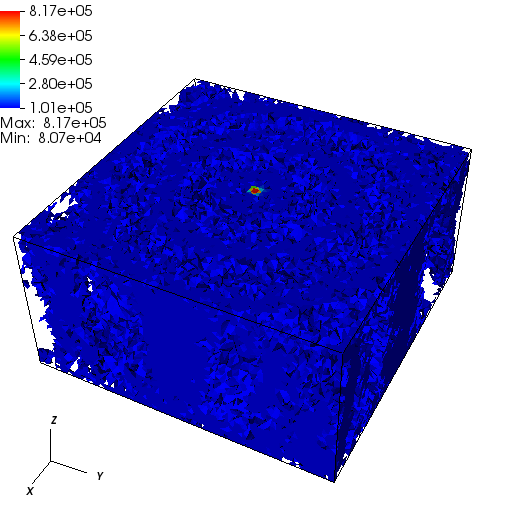}}
   	\subfigure[Proposed VI-based formulation]{\includegraphics[scale=0.29]{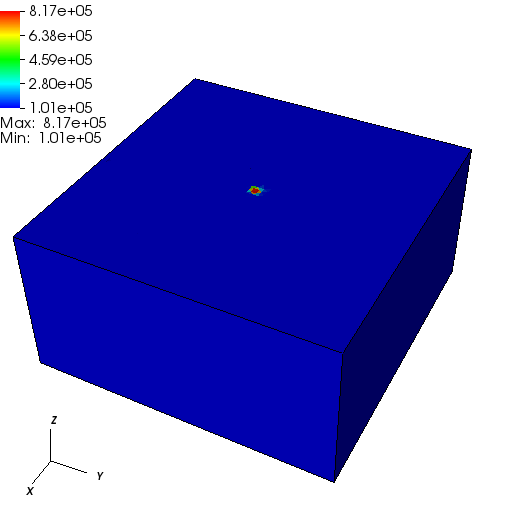}}
   	\captionsetup{format=hang}
	\vspace{-0.15in}
   	\caption*{\small{This picture shows the pressure profiles 
	of a 3D reservoir with a bore hole at the top surface. The 
	left figure depicts the pressure profile obtained using the 
	lowest-order Raviart-Thomas (RT0) formulation. The 
	missing chunks represent the regions in which the 
	discrete maximum principle (DMP) is violated. The 
	right figure shows the pressure profiles under the 
	proposed VI-based formulation, and there are no 
	violations of DMP.}}
   \end{figure}
    \vfill
    {\scshape 2017} \\
    {\small Computational \& Applied Mechanics Laboratory} \par
\end{titlepage}

\begin{abstract}
Mathematical models for flow through porous media typically 
enjoy the so-called maximum principles, which place bounds 
on the pressure field. 
It is highly desirable to preserve these bounds on the pressure field in predictive numerical simulations, that is, one needs to satisfy discrete maximum principles (DMP). Unfortunately, many of the existing formulations for flow through porous media models do \emph{not} satisfy DMP. 
  This paper presents a robust, scalable numerical formulation based on variational inequalities (VI), to model non-linear flows through heterogeneous, anisotropic porous media without violating DMP. VI is an optimization technique that places bounds on the numerical solutions of partial differential equations. To crystallize the ideas, a modification to Darcy equations by taking into account pressure-dependent viscosity will be discretized using the lowest-order Raviart-Thomas (RT0) and Variational Multi-scale (VMS) finite element formulations. It will be shown that these formulations violate DMP, and, in fact, these violations increase with an increase in anisotropy. It will be shown that the proposed VI-based formulation provides a viable route to enforce DMP. Moreover, it will be shown that the proposed formulation is scalable, and can work with any numerical discretization and weak form.
  A series of numerical benchmark problems are solved to demonstrate 
  the effects of heterogeneity, anisotropy and non-linearity on DMP 
  violations under the two chosen formulations (RT0 and VMS), and 
  that of non-linearity on solver convergence for the proposed 
  VI-based formulation.
  Parallel scalability on modern computational platforms will be 
  illustrated through strong-scaling studies, which will prove 
  the efficiency of the proposed formulation in a parallel setting. 
  Algorithmic scalability as the problem size is scaled
  up will be demonstrated through novel static-scaling 
  studies. The performed static-scaling studies can 
  serve as a guide for users to be able to select an 
  appropriate discretization for a given problem size.
\end{abstract}
\maketitle


\input{Sections/S1_NN_Intro}

\input{Sections/S2_NN_GE}

\input{Sections/S3_NN_Mixed}

\input{Sections/S4_NN_VI}

\input{Sections/S5_NN_NR}

\input{Sections/S6_NN_CR}

\bibliographystyle{plainnat}
\bibliography{References}
\end{document}

%% file: Sections/S1_NN_Intro.tex
\section{INTRODUCTION}
\label{Sec:S1_NN_Intro}

The success of many current and emerging technological
endeavors critically depend on a firm understanding
and on the ability to control flows in heterogeneous,
anisotropic porous media. These endeavors include
geological carbon sequestration, geothermal systems,
oil recovery, water purification systems, extraction
of gas hydrates from tight shale; just to name a few.
Modeling and predictive simulations play an important
role in all these endeavors, and one has to overcome
many numerical challenges to obtain accurate numerical
solutions. It is beyond the scope of this paper to address
all the major issues associated with the flow of fluids
through porous media. Herein, we however address one of
the main numerical challenges that is encountered in
numerical modeling of flow through porous media with
relevance to the mentioned applications. 

Flow through porous media models typically enjoy the
so-called maximum principles, which place bounds on
the pressure field. These bounds depend on the prescribed
data, which include boundary conditions, anisotropy of the
porous media, body force,
volumetric source, topology of the domain, and the
regularity of the boundary. The non-negative constraint
on the pressure (which basically implies the physical
condition that a fluid subject to a flow in a porous medium 
cannot sustain a ``suction" by itself) can be shown to be
a special case of the classical maximum principle. It is
imperative that these bounds on the pressure field are
preserved in a predictive numerical simulation; that
is, one needs to satisfy maximum principles in the
discrete setting. The discrete version of maximum
principles is commonly referred to as discrete
maximum principles (DMP). It becomes even more
crucial for those flow models in which 
the material properties depend on the pressure; for example,
the case in which the viscosity of the fluid depends on the
pressure in the fluid, as a violation of DMP can amplify
errors in the solution fields. Unfortunately, many of the
commonly used mixed finite element formulations
for flow through porous media models do \emph{not} satisfy 
DMP, which will be shown in the subsequent sections. Moreover,
the problems pertaining to flow through porous media,
especially the ones encountered in subsurface modeling,
are highly nonlinear and large-scale in nature. 
Thus, one needs to develop numerical formulations that are 
scalable in an algorithmic and parallel sense in 
addition to satisfying DMP. 

\emph{This paper presents a new, scalable numerical
  formulation based on variational inequalities (VI)
  that enforces discrete maximum principles for
  nonlinear flow through porous media models by taking into
  account heterogeneity, anisotropic permeability
  and pressure-dependent viscosity.}

\subsection{A review of related prior works}
In order to bring out clearly the contributions
made in this paper and the approach taken by us,
we provide a brief discussion on prior works with
respect to three aspects.

\subsubsection{Pressure-dependent viscosity.}
The classical Darcy model \citep{Darcy_1856},
which is the most popular flow through porous
media model, assumes the viscosity of the fluid 
to be a constant, and in particular, the model
assumes that the coefficient of viscosity is
independent of the pressure in the fluid
\citep{nakshatrala2011numerical}.
But there is abundant experimental evidence
that the viscosity of liquids, especially
organic liquids, depends on the pressure
\citep{Bridgman}. More importantly, the
dependence of viscosity on pressure for
organic liquids is exponential
\citep{CarlBarus}.
Since then several studies have developed mathematical
models that take into account the dependence of viscosity
on pressure, and established the existence of solutions
for the resulting governing equations \citep{malek2002global, 
  hron2003numerical, franta2005steady, bulivcek2007navier}. 
A work that is relevant to this paper is by 
\citep{nakshatrala2011numerical} who derived
a modification to the Darcy model using the
mixture theory by taking into account the
pressure-dependent viscosity. They have
also developed a stabilized formulation
for the resulting equations using the
variational multiscale paradigm
\citep{hughes1995multiscale}, and have shown,
using numerical simulations, that the dependence
of viscosity on pressure has a significant effect
on both qualitative and quantitative nature of the
solution fields.
Later, \citep{nakshatrala2013picard} have presented a
stabilized mixed formulation based on Picard linearization
and laid down the differences in the predictions for
enhanced oil recovery and carbon sequestration when
the pressure dependence on viscosity is considered
against when not considered.
Recently, \citep{Chang_JPM_2016} have extended the
pressure dependence to the Darcy-Forchheimer model,
and demonstrated how the dependence of the drag 
coefficient on pressure differs significantly
from when it depends on velocity. \emph{However,
  all of these studies considered isotropic
  permeability, and did not address the violations
  of maximum principles and the non-negative
  constraint on the pressure field.}

\subsubsection{Anisotropy, violations of DMP and numerical
  techniques to enforce DMP}
\citep{varga1966discrete} was the first to
address DMP, and the study was restricted
to the finite difference method applied on
the Poisson's equation.
\citep{ciarlet1973maximum} were the first
to address DMP in the context of the finite
element method. Their study revealed that
the single-field Galerkin formulation for
solving the Poisson's equation, in general,
does not satisfy DMP. They also obtained
sufficient conditions which are in the
form of restrictions on the mesh (e.g.,
all the angles of a triangular element to
be acute) to meet DMP. Subsequent studies
have found that these mesh restrictions,
which have been derived for isotropic
diffusion equations, are not sufficient
when one considers anisotropic diffusion
equations or other processes like advection
and reactions; for example, see
\citep{2015_Mudunuru_Nakshatrala_arXiv}
and references therein. 

In the last decade, several approaches have
been developed to enforce DMP on general
computational grids for anisotropic
diffusion-type equations under the finite
element method. Some of the notable
approaches are based on either constrained
optimization techniques \citep{Liska_CiCP_2008,
  nagarajan2011enforcing,mudunuru2015local},
placing anisotropic metric-based restrictions
on the mesh \citep{huang2015discrete,
  2015_Mudunuru_Nakshatrala_arXiv}, or
altering the formulations at the continuum
setting \citep{pal2016adaptive}.
Placing restrictions on meshes is not
a viable approach for applications
involving flow through porous media, as
the computational domains are complex
and it is not practical or even possible
to generate metric-based meshes that
satisfy DMP. The approach of altering
formulations at the continuum setting
is not a viable route either for porous
media applications, as one has to deal
with a hierarchy of models with multiple
constituents in such applications and there
is no trivial way of altering formulations
at the continuum level to meet DMP. 

Optimization-based methods based on quadratic
programming have been successfully employed
to develop formulations for anisotropic
diffusion equations, for example, see
\citep{NAKSHATRALA20096726}. The key behind
these methods has been to construct an objective
function in quadratic form; employ low-order
finite elements, whose shape functions are
non-negative within each element; and enforce
the constraints arising from DMP as explicit
bound constraints on the nodal quantities.
This approach has also been extended to transient
problems \citep{nakshatrala_nagarajan_shabouei_2016}, 
advection-diffusion equation \citep{mudunuru2016enforcing}, 
diffusion with fast reactions \citep{NAKSHATRALA2013278}, 
and parallel environments \citep{Chang_JOSC_2017}.
However, it needs to be emphasized that
this approach, which is based on quadratic
programming, requires the bilinear form in 
the weak form to be symmetric, which is not
the case with many porous media models and
weak formulations. 
\emph{More importantly, all the mentioned
  studies considered linear equations
  arising from transport problems.}  

\subsubsection{Variational inequality based techniques.}
Variational inequalities arise quite naturally in
various branches of mechanics \citep{kikuchi1988contact,
  hlavacek2012solution,rodrigues1987obstacle,han2012plasticity}.
In fact, the whole field of variational inequalities grew
from a problem in mechanics which was posed by
\citep{signorini1933sopra}. This problem, which is
popularly referred to as the Signorini problem, is
about finding static equilibrium configurations
of a linear elastic body resting on a rigid smooth
surface \citep{signorini1959questioni}.
A work that is more directly related to this paper
is by \citep{chipot2012variational}, who has employed
variational inequalities to address some class of
problems that arise in studies on flow through porous
media, specifically the dam problem, and established 
mathematical properties like existence and uniqueness
of solutions. 
The Signorini problem and the treatment of the dam problem
are examples of infinite-dimensional variational inequalities.
Subsequently the field of finite-dimensional variational
inequalities has been developed \citep{Facchinei_FDVI_2003},
and this field has eventually found its way into the mainstream
numerical optimization \citep{Kinderlehrer_VI_2000,ulbrich2011semismooth}.
Recently, finite-dimensional variational inequalities have been
utilized to enforce DMP and the non-negativity of concentrations
under advection-diffusion equations \citep{chang2017variational}. 
This formulation does not require the bilinear form of 
the underlying weak formulation to be symmetric. 
However, it needs to be emphasized that advection-diffusion
equations, which arise in transport problems, are linear
and are in terms of a single unknown field. In this paper,
we address flow problems, and the governing equations are
in terms of two fields (i.e., the velocity and pressure)
and are nonlinear. 

\subsection{Our approach, salient features and an outline of the paper}
We develop a numerical formulation based on
variational inequalities for modeling flow
through porous media that accounts for anisotropy
and pressure-dependent viscosity, and possess the
following attractive features:
\begin{enumerate}[(i)]
\item The proposed framework can handle any
  numerical discretization and weak formulation.
\item The devised computational framework is
  equipped to handle non-linear formulations
  and problems with non-self-adjoint operators.
\item Maximum principles are satisfied under this
  framework even when anisotropy is present.
\item A computer implementation of the
  proposed framework can seamlessly
  leverage on the state-of-the-art
  software and algorithms that are 
  currently available for high
  performance computing. 
\item The implementation outlined in
  this paper has excellent scalability
  both in the algorithmic and parallel
  sense.
\end{enumerate}
All the aforementioned features of the
proposed formulation will be illustrated
in the subsequent sections. 
The rest of this paper is organized as follows. In 
Section \ref{Sec:S2_NN_GE} we present the governing equations for 
the modified Darcy flow and describe discrete maximum principles. 
In Section \ref{Sec:S3_NN_Mixed} we provide various mixed and 
nonlinear formulations used for this study. In Section \ref{Sec:S4_NN_VI}, 
we lay down the solver methodology and outline of the computer implementation 
for the proposed VI based framework. In Section \ref{Sec:S5_NN_NR}, we present 
numerical results illustrating effectiveness and scalability of the 
proposed framework. Concluding remarks are made in 
Section \ref{Sec:S6_NN_CR}.

%% file: Sections/S2_NN_GE.tex
\section{A NON-LINEAR MODEL FOR FLOW THROUGH POROUS MEDIA}
\label{Sec:S2_NN_GE}
Consider a porous domain denoted by
$\Omega \subset \mathbb{R}^{nd}$, where
``$nd$'' denotes the number of spatial
dimensions. The boundary
of the domain will be denoted by $\partial \Omega :=
\overline{\Omega} - \Omega$, where an overline
denotes the set closure. A spatial point will be 
denoted by $\mathbf{x} \in \overline{\Omega}$.
The divergence and gradient operators with
respect to $\mathbf{x}$ are, respectively,
denoted by $\mathrm{div}[\cdot]$ and
$\mathrm{grad}[\cdot]$. The unit outward normal
to the boundary is denoted by $\widehat{\mathbf{n}}$.
The discharge velocity vector field and
the pressure scalar field are denoted by
$\mathbf{u}$ and $p$, respectively.
The boundary is divided into two parts:~$\Gamma^{u}$
and $\Gamma^{p}$, such that we have 
\begin{align}
  \Gamma^{u} \cup \Gamma^{p} = \partial\Omega
  \quad \mathrm{and} \quad
  \Gamma^{u} \cap \Gamma^{p} = \emptyset.
\end{align}
$\Gamma^{u}$ denotes that part of the boundary on
which the normal component of the velocity is
prescribed. $\Gamma^\mathit{p}$ is that part of
the boundary on which pressure is prescribed.
The permeability of the porous medium
will be denoted by $\mathbf{K}$, which
is a second-order tensor. It is assumed
that the permeability tensor is positive
definite and symmetric. The density of
the fluid is denoted by $\rho$. The
coefficient of (dynamic) viscosity of
the fluid is denoted by $\mu$.

As mentioned earlier, for most organic
liquids, the dependence of pressure on
viscosity is exponential. That is,
mathematically we have 
\begin{align}
  \label{eq:barus_exp}
  \mu(p) = \mu_0 \exp[\beta_{B} p], 
\end{align}
where $\beta_B$ is the Barus coefficient,
which has to be obtained experimentally
and the values of this coefficient for
various liquids can be found in
\citep{Bridgman}. Since the pressures that 
we deal in this paper are relatively small 
but not sufficiently small enough to neglect the 
pressure dependence of viscosity, we take a two-term 
Taylor expansion of the Barus formula (given
by equation \eqref{eq:barus_exp}). The two term Taylor
expansion, which will be employed in all the numerical
experiments in this paper, takes the following mathematical
form: 
\begin{align}
  \label{eq:barus_lin}
  \mu(p) = \mu_{0} \left(1 + \beta_{B}p \right). 
\end{align}

For convenience, we introduce the drag
coefficient, which is defined as follows:
\begin{alignat}{1}
  \alpha(p) = \mu(p) \mathbf{K}^{-1}.
\end{alignat}
Since the viscosity depends on the pressure and
the permeability explicitly depends on the spatial
coordinates, the drag coefficient will explicitly
depend on both the pressure and $\mathbf{x}$. 

The governing equations for flow through
porous media by taking into account the
pressure-dependent viscosity can be
written as follows:
\begin{subequations}\label{eq:modified_darcy}
  \begin{alignat}{2}
    \label{eq:modified_darcy_LM}
    & \alpha(p) \mathbf{u}
    + \mathrm{grad}[p]
    = \rho \mathbf{b}
    && \quad \mathrm{in} \quad \Omega, \\
    \label{eq:modified_darcy_BoM}
    &\mathrm{div}[\mathbf{u}]= f
    && \quad \mathrm{in} \quad \Omega, \\
    &\mathbf{u} \cdot \widehat{\mathbf{n}}
    = u_{n}
    && \quad \mathrm{on} \quad \Gamma^u, \; \mathrm{and}\\
    \label{eq:modified_darcy_pBC}
    &p = p_{0}
    && \quad \mathrm{on} \quad \Gamma^p, 
  \end{alignat}
\end{subequations}
where $u_{n}$ denotes the prescribed normal
component of velocity on the boundary, $p_{0}$
denotes the prescribed pressure on the boundary,
$\mathbf{b}$ is the specific body force, and $f$
is the prescribed volumetric source, all of which
are functions of $\mathbf{x}$.
It should be noted that one can recover
the classical Darcy equations by setting
$\beta_{B} = 0$, which makes the drag
coefficient to be independent of the
pressure. 
A systematic derivation of the above
governing equations under the theory
of interacting continua can be found
in \citep{nakshatrala2011numerical}.

\subsection{A mathematical interlude}
For a mathematical treatment of the abstract
boundary value problem \eqref{eq:modified_darcy},
we assume to have pressure boundary conditions on
the entire boundary (i.e., $\Gamma^{p} = \partial
\Omega$). We also rewrite the governing equations
solely in terms of the pressure as follows:
\begin{subequations}
  \begin{alignat}{2}
    \label{Eqn:VI_single_field_GE}
    -&\mathrm{div}\left[\frac{1}{\alpha(p)}
      (\mathrm{grad}[p] - \rho \mathbf{b}) 
      \right] = f
    &&\quad \mathrm{in} \; \Omega \; \mathrm{and} \\
    \label{Eqn:VI_single_field_BC}
    &p = p_{0} 
    &&\quad \mathrm{on} \; \partial \Omega. 
  \end{alignat}
\end{subequations}
Note that equation \eqref{Eqn:VI_single_field_GE} is
obtained by combining equations \eqref{eq:modified_darcy_LM}
and \eqref{eq:modified_darcy_BoM}. 
Equation \eqref{Eqn:VI_single_field_GE} is
a special case of a second-order quasi-linear
elliptic partial differential equation. 

A general second-order quasi-linear elliptic
operator takes the following form:
\begin{align}
  \label{Eqn:VI_general_quasilinear}
  Q[u] = \mathrm{div}[\mathbf{A}(\mathbf{x},u,\mathrm{grad}[u])]
  + B(\mathbf{x},u,\mathrm{grad}[u]). 
\end{align}
We define the coefficient matrix as follows:
\begin{align}
  \mathcal{A}(\mathbf{x},u,\mathbf{h})
  = \mathrm{sym}\left[\frac{\partial
      \mathbf{A}(\mathbf{x},u,\mathbf{h})}{\partial \mathbf{h}}\right],
\end{align}
which, in indicial notation, takes the following form:
\begin{align}
  \mathcal{A}_{ij} = \frac{1}{2} \left(
  \frac{\partial A_{i}}{\partial h_{j}}
  + \frac{\partial A_{j}}{\partial h_{i}}\right).
\end{align}
Note that the entries of the coefficient matrix
need not be constants. 
The operator $Q[u]$ is said to be elliptic
in $\Omega$ if the coefficient matrix
$\mathcal{A}(\mathbf{x},u,\mathbf{h})$
is positive definite for all $\mathbf{x}
\in \Omega$, $u \in \mathbb{R}$ and
$\mathbf{h} \in \mathbb{R}^{nd}$. 
From the theory of partial differential
equations, this operator is known to
satisfy the following important property:
\begin{theorem}{(Comparison principle in a general setting)}
  \label{Thm:VI_CP}
  Let $u, v \in C^{1}(\overline{\Omega})$ which
  satisfy $Q[u] \geq 0$ and $Q[v] \leq 0$ in
  $\Omega$ and $u \leq v$ on $\partial \Omega$.
  If the following conditions are met:
  \begin{enumerate}[(i)]
  \item $\mathbf{A}(\mathbf{x},u,\mathbf{h})$ and
    $B(\mathbf{x},u,\mathbf{h})$ are continuously
    differentiable with respect to $u$ and
    $\mathbf{h}$, 
  \item $Q[u]$ is elliptic in $\Omega$, 
  \item $B(\mathbf{x},u,\mathbf{h})$ is
    non-increasing with respect to $u$
    for fixed $(\mathbf{x},\mathbf{h})$, and
  \item $\partial B(\mathbf{x},u,\mathbf{h})
    /\partial \mathbf{h} = \mathbf{0}$ 
  \end{enumerate}
  then we
  have $u \leq v$ in $\Omega$. 
\end{theorem}
\begin{proof}
  A mathematical proof can be found in
  \citep[Theorem 10.7]{Gilbarg_Trudinger}.
\end{proof}

We now show that the solutions of the porous media
model satisfy the non-negative constraint and the
maximum principle. To this end, we assume that
the body force is conservative. That is, there
exists a scalar field $\psi$ such that
\begin{align}
  \rho \mathbf{b} = -\mathrm{grad}[\psi].
\end{align}

\begin{lemma}
  \label{Lemma:VI_properties}
  The porous media model, given by equations  
  \eqref{Eqn:VI_single_field_GE}, along with
  the pressure dependent viscosity given
  by equation \eqref{eq:barus_exp}, can be put into
  the following form:
  \begin{align}
    Q[p] + f = 0 \quad \mathrm{in} \; \Omega, 
  \end{align}
  with the following properties: 
  \begin{enumerate}[(i)]
  \item $\mathbf{A}(\mathbf{x},u,\mathbf{h})$
    and $B(\mathbf{x},u,\mathbf{h})$ are
    continuously differentiable with respect
    to $u$ and $\mathbf{h}$, 
  \item $Q[u]$ is elliptic in $\Omega$, 
  \item $B(\mathbf{x},u,\mathbf{h})$ is
    non-increasing with respect to $u$, and 
  \item $\partial B(\mathbf{x},u,\mathbf{h})/
    \partial \mathbf{h} = \mathbf{0}$. 
  \end{enumerate}
\end{lemma}
\begin{proof}
  It is a straightforward computation to
  show that with the following choices:
  \begin{subequations}
    \begin{align}
      \mathbf{A}(\mathbf{x},u,\mathrm{grad}[u])
      &= \frac{1}{\alpha(u)}\left(\mathrm{grad}[u]
      - \rho \mathbf{b}\right)
      = \frac{1}{\alpha(u)}\mathrm{grad}[u + \psi] \; \mathrm{and} \\
      B(\mathbf{x},u,\mathrm{grad}[u]) & = 0, 
    \end{align}
  \end{subequations}
  equation \eqref{Eqn:VI_single_field_GE}
  can be written as
  \begin{align}
    Q[p] + f = 0 \quad \mathrm{in} \; \Omega. 
  \end{align}
  $B(\mathbf{x},u,\mathbf{h}) = 0$ implies
  that conditions (iii), (iv) and the second
  part of (i) are trivially satisfied. Using
  equation \eqref{eq:barus_exp} we have
  \begin{align}
    \mathbf{A}(\mathbf{x},u,\mathbf{h})
    = \frac{1}{\mu_0} \mathbf{K} \exp[-\beta_{B}
      u] \left(\mathbf{h} + \mathrm{grad}[\psi]\right).
  \end{align}
  Clearly, $\mathbf{A}(\mathbf{x},u,\mathbf{h})$
  is continuously differentiable with respect to
  $u$ and $\mathbf{h}$, which implies that the
  first part of condition (i) is satisfied. The
  coefficient matrix can be written as follows:
  \begin{align}
    \mathcal{A}(\mathbf{x},u,\mathbf{h})
    = \frac{1}{\mu_0} \mathbf{K}
    \exp[-\beta_{B} u].
  \end{align}
  The positive definiteness of the permeability
  tensor, $\mathbf{K}$, and $\mu_0 > 0$ imply
  that the coefficient matrix is positive
  definite for all $\mathbf{x} \in \Omega$,
  $u \in \mathbb{R}$ and $\mathbf{h} \in
  \mathbb{R}^{nd}$; which further implies that
  condition (ii) is met. 
\end{proof}

\begin{theorem}{(Non-negative pressures under the porous media model)}
  \label{Thm:VI_non-negative}
  Let $\psi$ be sufficiently smooth and $p \in C^{1}(\overline{\Omega})$. 
  If the prescribed volumetric source is non-negative
  in $\Omega$ and the prescribed pressure on the
  boundary is non-negative then the pressure in
  the entire domain is non-negative. That is, if
  $f \leq 0$ in $\Omega$ and $p_0 \geq 0$ on
  $\partial \Omega$ then $p \geq 0$ in $\Omega$. 
\end{theorem}
\begin{proof}
  Choose $u = -p -\psi$ and $v = - \psi$. We then have
  \begin{align}
    Q[u] = - f \geq 0 \quad \mathrm{and} \quad Q[v] = 0. 
  \end{align}
  If $p_0 \geq 0$ on $\partial \Omega$ we have
  \begin{align}
    u = -p -\psi = -p_0 -\psi \leq -\psi = v
    \quad \mathrm{on} \; \partial \Omega. 
  \end{align}
  Using Lemma \ref{Lemma:VI_properties} and the comparison principle given by
  Theorem \ref{Thm:VI_CP}, we conclude that
  \begin{align}
    u \leq v \quad \mathrm{in} \; \Omega. 
  \end{align}
  This further implies that
  \begin{align}
    0 = u + \psi \leq v + \psi = p
    \quad \mathrm{in} \; \Omega, 
  \end{align}
  which completes the proof.
\end{proof}

\begin{theorem}{(Maximum principle for the porous media model)}
  \label{Thm:VI_maximum_principle}
  If the prescribed volumetric source is zero
  (i.e., $f = 0$) in $\Omega$ then the maximum
  and minimum pressures occur on the boundary.
  That is,
  \begin{align}
    \min\left[p_{0}\right]
    \leq p(\mathbf{x}) \leq
    \max\left[p_{0}\right]
    \quad \forall \mathbf{x}
    \in \overline{\Omega}. 
  \end{align}
\end{theorem}
\begin{proof}
  To show the right inequality, we take 
  \begin{align}
    u = -\psi -\max[p_0] \quad \mathrm{and} \quad
    v = -\psi - p . 
  \end{align}
  Since $p \leq \max[p_0]$ on $\partial \Omega$,
  we have $u \leq v$ on $\partial \Omega$.
  Moreover, the above choices for $u$
  and $v$ imply that 
  \begin{align}
    Q[u] = 0
    \quad \mathrm{and} \quad
    Q[v] = Q[p] = f = 0
    \quad \mathrm{in} \; \Omega. 
  \end{align}
  Lemma \ref{Lemma:VI_properties} and the
  comparison principle given by Theorem
  \ref{Thm:VI_CP} imply that 
  \begin{align}
    u \leq v \quad \mathrm{in} \; \Omega. 
  \end{align}
  This further implies that 
  \begin{align}
    -\psi - \max[p_0] = u \leq v = -\psi - p
    \quad \mathrm{in} \; \Omega. 
  \end{align}
  One can thus conclude that $p \leq \max[p_0]$.

  To show the left inequality, we take
  \begin{align}
    u = -\psi -p \quad \mathrm{and} \quad
    v = -\psi -\min[p_0], 
  \end{align}
  which imply that
  \begin{align}
    &u \leq v \quad \mathrm{on} \; \partial \Omega, \\
    &Q[u] \geq 0 \; \mathrm{and} \;
    Q[v] \leq 0 \quad \mathrm{in} \; \Omega. 
  \end{align}
  By again appealing to the comparison principle,
  we conclude that $\min[p_0] \leq p$ in $\Omega$. 
\end{proof}
The existing numerical discretizations for flow
through porous media models do not produce
solutions that satisfy the aforementioned
mathematical properties for anisotropic porous domains.
Thus, the central aim of this paper is to develop
a computational framework for nonlinear models
for flow through porous media that satisfies
the maximum principle and ensures non-negative
solutions for the pressure. This will be achieved
by combining mixed finite element methods and
variational inequalities.

%% file: Sections/S3_NN_Mixed.tex
\section{MIXED FORMULATIONS}
\label{Sec:S3_NN_Mixed}
In our study, we employ two well-established
finite element formulations which achieve
discrete stability differently. The stability
of a mixed formulation in the discrete setting
will be primarily dictated by the famous
Ladyzhenskaya-Bab{\v u}ska-Brezzi stability
condition \citep{babuvska1973finite,brezzi1991mixed}.
The first formulation is the classical mixed 
formulation (also known as the Galerkin weak 
formulation) but the interpolations for the 
velocity and pressure fields are based on 
the lowest-order Raviart-Thomas space
\citep{Raviart1977}. 
It is well-known that an arbitrary combination
of interpolation functions for the velocity
and pressure fields under the classical
mixed formulation need not satisfy the
LBB condition, and hence may not be stable
\citep{brezzi1991mixed,brezzi2008mixed}.
The Raviart-Thomas spaces place restrictions
on the interpolations for the velocity and
pressures fields to \emph{satisfy} the LBB
condition, and thus making the classical
mixed formulation stable.
The second formulation is the Variational
Multi-scale formulation \citep{nakshatrala2011numerical},
which is a stabilized mixed formulation that
augments the Galerkin weak formulation
with stabilization terms to \emph{circumvent}
the LBB condition. A nice discussion on the
two classes of mixed formulation, which differ
in the way they handle the LBB condition
(i.e., satisfying vs. circumventing),
can be found in \citep{franca1988two}.

The weak forms of the aforementioned two formulations
will form the basis for the associated variational
inequalities. To this end, the following function 
spaces will be employed in the rest of the paper:
\begin{subequations}
  \begin{align}
   \mathcal{U} &:= \left\{ \mathbf{u} \in \left(L_2(\Omega)\right)^{nd}
    \; | \; \mathrm{div}[\mathbf{u}] \in L_2(\Omega), \; 
    \mathbf{u} \cdot \widehat{\mathbf{n}} = u_{n} \; \mathrm{on} \;
    \Gamma^u  \right\}, \\
    \mathcal{W} &:= \left\{ \mathbf{w} \in \left(L_2(\Omega)\right)^{nd}
    \; | \; \mathrm{div}[\mathbf{w}] \in L_2(\Omega), \; 
    \mathbf{w} \cdot \widehat{\mathbf{n}} = 0 \; \mathrm{on} \;
    \Gamma^u  \right\}, \\
    \mathcal{P} &:= L_2(\Omega), \; \mathrm{and} \\ 
    \mathcal{Q} &:= H^1(\Omega),
  \end{align}
\end{subequations}
where $H^{1}(\Omega)$ is a standard Sobolov space 
\citep{brezzi1991mixed} and $L_2(\Omega)$ is set 
of all square integrable functions on $\Omega$. 
The Galerkin weak formulation for the governing 
equations \eqref{eq:modified_darcy} reads: Find 
$\mathbf{u} \in \mathcal{U}$ and $p \in \mathcal{P}$ 
such that we have
\begin{align}
  \label{eq: Gal}
  &\int_\Omega  \alpha(p) \mathbf{u} \cdot \mathbf{w} 
  \;\mathrm{d}\Omega - \int_\Omega \mathrm{div}[\mathbf{w}] 
  \; p \;\mathrm{d}\Omega  - \int_\Omega \mathrm{div}[\mathbf{u}]
  \; q\;\mathrm{d}\Omega \nonumber \\
  &\qquad= \int_{\Omega} \rho\mathbf{b} \cdot \mathbf{w} \; 
  \mathrm{d}\Omega -\int_{\Omega} f q\; \mathrm{d}\Omega 
  - \int_{\Gamma^\mathrm{P}} p_0 (\mathbf{w}\cdot \widehat{\mathbf{n}}) 
  \; \mathrm{d}\Gamma 
  \quad \forall \mathbf{w} \in \mathcal{W}, q \in \mathcal{P}.
\end{align}

\subsection{Lowest-order Raviart-Thomas space}
Given a simplex $\mathcal{T} \in \mathbb{R}^{nd}$, 
the local Raviart-Thomas space of order $k \geq 0$ 
is defined as follows \citep{Raviart1977,bergamaschi1994mixed}:
\begin{align}
  \mathcal{RT}_k\mathcal{(T)} 
  = \left(\mathcal{P}_k\mathcal{(T)}\right)^{nd} 
  + \mathbf{x}\mathcal{P}_k\mathcal{(T)} 
\end{align}
where $ \mathcal{P}_k$ is the space of polynomials 
of degree $k$ and $nd$, as mentioned before, is the 
number of spatial dimensions. 
It is well-known that the interpolations for 
the velocity and pressure fields under the 
Raviart-Thomas spaces of all orders satisfy 
the LBB \emph{inf-sup} stability condition 
and thereby provide stable numerical solutions 
under the Galerkin weak formulation 
\citep{brezzi1991mixed}.

In this study we employ the lowest-order Raviart-Thomas 
space for interpolation of velocity and pressure fields, 
which is the simplest and the most popular space among 
the class of Raviart-Thomas spaces.
Under the lowest-order Raviart-Thomas space, the 
pressure is constant within an element and the 
fluxes are evaluated at the midpoint of each 
edge in 2D or at the barycenter of each face 
in 3D. Mathematically, 
\begin{align}
  \mathcal{RT}_0\mathcal{(T)} = \left(\mathcal{P}_0\mathcal{(T)}\right)^{nd} 
  + \mathbf{x}\mathcal{P}_0\mathcal{(T)}.
\end{align}
The finite dimensional subspaces $\mathcal{U}^h 
\subset \mathcal{U}$ and $\mathcal{P}^h \subset 
\mathcal{P}$ under $\mathcal{RT}_0$ for a triangle 
are defined as follows: 
\begin{subequations}
  \begin{align}
    \mathcal{U}^h &:= \{\mathbf{u}=(u^{(1)},u^{(2)})\;|\;u^{(1)}_K = a_K + b_K x,\; u^{(2)}_K 
    = c_K + b_K y;\; a_K,b_K,c_K \in \mathbb{R}\} \; \mathrm{and}  \\ 
    \mathcal{P}^h &:= \left\{p\;|\;p\;=\;\mbox{a constant on each triangle }\; K \in \mathcal{T}_h  \right\}, 
  \end{align}
\end{subequations}
where $ \mathcal{T}_h $ is a triangulation on 
$\Omega $. These subspaces for tetrahedra are 
defined as follows:
\begin{subequations}
  \begin{align}
    \mathcal{U}^h &:= \{ \mathbf{u}=(u^{(1)},u^{(2)},u^{(3)})\; | \nonumber \\
    & u^{(1)}_i = a_i + b_i x,\; u^{(2)}_i 
    = c_i + b_i y,\;
    u^{(3)}_i = d_i + b_i z; \; a_i,b_i,c_i,d_i \in \mathbb{R}\} \; \mathrm{and}  \\ 
    \mathcal{P}^h &:= \left\{p\;|\;p\;=\;\mbox{constant on each tetrahedron}\; K \in \mathcal{T}_h  \right\}. 
	\end{align}
\end{subequations}
where $ \mathcal{T}_h $, in this case, is a 
tetrahedralization (i.e., 3D triangulation) 
on $\Omega $. 

\subsection{Variational Multi-scale formulation}
Variational Multi-scale (VMS) is a computational paradigm 
to achieve enhanced stability of a given weak formulation 
\citep{hughes1995multiscale}. For a mixed formulation, 
say the Galerkin formulation, residual-based adjoint-type 
stabilization terms are added to circumvent the LBB 
condition and achieve stability. 
\citep{nakshatrala2011numerical} have successfully 
employed the VMS paradigm to develop a stabilized 
mixed formulation for the isotropic version of 
the porous media model outlined in Section 
\ref{Sec:S2_NN_GE}. 
The weak form under the VMS formulation 
for governing equations 
\eqref{eq:modified_darcy_LM}--\eqref{eq:modified_darcy_pBC} 
reads: Find $\mathbf{u} \in \mathcal{U}$ and $p \in \mathcal{P}$ 
such that we have 
\begin{align}\label{eq: VMS}\nonumber
  \int_\Omega  \mathbf{\alpha}(p) \mathbf{u} \cdot \mathbf{w}  & \;\mathrm{d}\Omega 
  - \int_\Omega p \; \mathrm{div}[\mathbf{w}] \;\mathrm{d} \Omega 
  - \int_\Omega q \; \mathrm{div}[\mathbf{u}] \;\mathrm{d} \Omega \\ \nonumber
  &\underbrace{-\frac{1}{2}\int_\Omega (\mathbf{u} + \mathbf{\alpha}^{-1}(p) \mathrm{grad}[p])\cdot (\mathbf{\alpha}(p) \mathbf{w} + \mathrm{grad}[q]) \;\mathrm{d}\Omega} _{\mbox{stabilization term}}  \\ \nonumber
  = \int_{\Omega} \rho \mathbf{b} \cdot \mathbf{w}  \;\mathrm{d}\Omega  
  &-\int_{\Omega} q \; f \;\mathrm{d} \Omega 
  - \int_{\Gamma^P} p_0 (\mathbf{w}\cdot \widehat{\mathbf{n}}) \; \mathrm{d}\Gamma \\ 
  & \underbrace{-\frac{1}{2}\int_{\Omega}\mathbf{\alpha}^{-1}(p) \rho \mathbf{b}\cdot (\mathbf{\alpha}(p)\mathbf{w}+\mathrm{grad}[q]) \; \mathrm{d}\Omega}_{\mbox{stabilization term}} 
\quad \forall \mathbf{w} \in \mathcal{W}, q \in \mathcal{Q}. 
\end{align}
In all our numerical simulations, we employ 
equal-order linear nodal-based interpolations 
for the pressure and velocity fields.  

\subsection{Non-linear formulations}
The pressure dependence of viscosity in the weak
formulations turns the problem into a non-linear
problem. To solve such problems, we introduce the
canonical form:~Find $\mathbf{u} \in \mathcal{U}$
and $p \in \mathcal{Q}$ such that we have 
\begin{align}
  \label{Eqn:S3_RT0_canonical}
  &\mathcal{F}\left[(\mathbf{u},p);(\mathbf{w},q)\right] 
  = 0 \quad \forall \mathbf{w} \in \mathcal{W}, \;
  \forall q \in \mathcal{Q},
\end{align}
where $\mathcal{F}$ is the residual expressed in 
semi-linear form; the arguments to the left and 
right of the semicolon are non-linear and linear, 
respectively. The semi-linear form for the RT0 
formulation takes the following form:
\begin{align}
  \label{Eqn:S3_RT0_residual}
  &\mathcal{F}_{\mathrm{RT0}}\left[(\mathbf{u},p);(\mathbf{w},q)\right] := 
    \int_{\Omega} \alpha(p)\mathbf{u}\cdot\mathbf{w}\;\mathrm{d} \Omega 
    - \int_{\Omega} p\cdot\mathrm{div}[\mathbf{w}]\;\mathrm{d}\Omega
    - \int_{\Omega} \mathrm{div}\left[\mathbf{u}\right]\cdot q\;\mathrm{d}\Omega  \nonumber \\
    &\qquad+ \int_{\Gamma^{\mathrm{P}}} p_0 \cdot \left(\mathbf{w}\cdot\mathbf{\hat{n}}\right)
     \; \mathrm{d} \Gamma + \int_{\Omega}q\cdot f\;\mathrm{d}\Omega 
     - \int_{\Omega}\rho\mathbf{b}\cdot\mathbf{w}\;\mathrm{d} \Omega.
\end{align}
The semi-linear form for the VMS formulation
can be written as follows: 
\begin{align}
  \label{Eqn:S3_VMS_residual}
  &\mathcal{F}_{\mathrm{VMS}}\left[(\mathbf{u},p);(\mathbf{w},q)\right] := 
  \frac{1}{2}\int_{\Omega} \alpha(p)\mathbf{u}\cdot\mathbf{w}\;\mathrm{d} \Omega 
    - \int_{\Omega} p \; \mathrm{div}[\mathbf{w}]\;\mathrm{d}\Omega
    - \int_{\Omega} \mathrm{div}\left[\mathbf{u}\right]\cdot q\;\mathrm{d}\Omega \nonumber\\
    &\qquad- \frac{1}{2}\int_{\Omega}\mathbf{u}\cdot\mathrm{grad}[q]\;\mathrm{d}\Omega 
    - \frac{1}{2}\int_{\Omega}\mathrm{grad}[p]\cdot\mathbf{w}\;\mathrm{d}\Omega 
    + \int_{\Gamma^{\mathrm{P}}} p_0 \left(\mathbf{w}\cdot\mathbf{\hat{n}}\right)
     \; \mathrm{d} \Gamma \nonumber\\
     &\qquad- \frac{1}{2}\int_{\Omega}\alpha^{-1}(p)\left(\mathrm{grad}[p]-\rho\mathbf{b}\right)\cdot\mathrm{grad}[q]\;\mathrm{d}\Omega
     - \frac{1}{2}\int_{\Omega}\rho\mathbf{b}\cdot\mathbf{w}\;\mathrm{d} \Omega\nonumber\\
    &\qquad+ \int_{\Omega}q\cdot f\;\mathrm{d}\Omega. 
\end{align}
Newton's method is employed to solve the
non-linear variational forms. Let the
superscript ($i$) denote the current Newton 
or non-linear iteration. The Jacobian $\mathcal{J}\left[\left(\mathbf{u}^{(i)},p^{(i)}\right);
\left(\delta\mathbf{u},\delta p\right),\left(\mathbf{w},q\right)\right]$
is computed by taking the G$\mathrm{\hat{a}}$teaux variation of
the residual $\mathcal{F}\left[\left(\mathbf{u},p\right),
\left(\mathbf{w},q\right)\right]$ at $\mathbf{u} = \mathbf{u}
^{(i)}$ and $p = p^{(i)}$ in the directions of
$\delta\mathbf{u}$ and $\delta p$ respectively. Formally, this is derived by
computing:
\begin{align}
\label{Eqn:NL_gateaux}
&\mathcal{J}\left[\left(\mathbf{u}^{(i)},p^{(i)}\right);\left(\delta\mathbf{u},\delta p\right),\left(\mathbf{w},q\right)\right] := \nonumber\\
&\qquad \mathop{\mathrm{lim}}_{\epsilon\rightarrow 0}
\frac{\mathcal{F}\left[\left(\mathbf{u}^{(i)}+\epsilon\delta\mathbf{u},p^{(i)}+\epsilon\delta p\right);\left(\mathbf{w},q\right)\right] -
\mathcal{F}\left[\left(\mathbf{u}^{(i)},p^{(i)}\right);\left(\mathbf{w},q\right)\right]}{\epsilon} \nonumber\\
&\qquad\equiv
\left[\frac{d}{d\epsilon}\mathcal{F}\left[\left(\mathbf{u}^{(i)}
+ \epsilon\delta\mathbf{u},p^{(i)}
+ \epsilon\delta p\right);\left(\mathbf{w},q\right)\right]\right]_{\epsilon=0},
\end{align}
provided the limit exists. For further details
on the G\^{a}teaux variation see
\citep{Spivak, Holzapfel, Glowinski}. Following through
with the calculation above yields the following Jacobian
under the RT0 formulation:
\begin{align}
\label{Eqn:S3_Jacobian_RT0}
&\mathcal{J}_{\mathrm{RT0}}\left[\left(\mathbf{u}^{(i)},p^{(i)}\right);\left(\delta\mathbf{u},\delta p\right),\left(\mathbf{w},q\right)\right] :=  \int_{\Omega} \alpha(p^{(i)})\delta\mathbf{u}\cdot\mathbf{w}\;\mathrm{d} \Omega 
+ \int_{\Omega} \frac{\partial \alpha(p^{(i)})}{\partial p}\mathbf{u}^{(i)}\delta p\cdot\mathbf{w}\;\mathrm{d} \Omega \nonumber\\
&\qquad- \int_{\Omega} \delta p\cdot\mathrm{div}[\mathbf{w}]\;\mathrm{d}\Omega 
- \int_{\Omega} \mathrm{div}\left[\delta\mathbf{u}\right]\cdot q\;\mathrm{d}\Omega. 
\end{align}
Likewise, the Jacobian for the VMS formulation reads:
\begin{align}
\label{Eqn:S3_Jacobian_VMS}
&\mathcal{J}_{\mathrm{VMS}}\left[\left(\mathbf{u}^{(i)},p^{(i)}\right);\left(\delta\mathbf{u},\delta p\right),\left(\mathbf{w},q\right)\right] := 
\frac{1}{2}\int_{\Omega} \alpha(p^{(i)})\delta\mathbf{u}\cdot\mathbf{w}\;\mathrm{d} \Omega 
+ \frac{1}{2}\int_{\Omega} \frac{\partial \alpha(p^{(i)})}{\partial p}\mathbf{u}^{(i)}\delta p\cdot\mathbf{w}\;\mathrm{d} \Omega \nonumber\\
&\qquad- \int_{\Omega} \delta p\cdot\mathrm{div}[\mathbf{w}]\;\mathrm{d}\Omega 
- \int_{\Omega} \mathrm{div}\left[\delta\mathbf{u}\right]\cdot q\;\mathrm{d}\Omega 
- \frac{1}{2}\int_{\Omega}\delta\mathbf{u}\cdot\mathrm{grad}[q]\;\mathrm{d}\Omega 
- \frac{1}{2}\int_{\Omega}\mathrm{grad}[\delta p]\cdot\mathbf{w}\;\mathrm{d}\Omega \nonumber \\
&\qquad- \frac{1}{2}\int_{\Omega}\alpha^{-2}(p^{(i)})\frac{\partial \alpha(p^{(i)})}{\partial p}\left(\mathrm{grad}\left[p^{(i)}\right]-\rho\mathbf{b}\right)\delta p\cdot\mathrm{grad}[q]\;\mathrm{d}\Omega\nonumber\\
&\qquad- \frac{1}{2}\int_{\Omega}\alpha^{-1}(p^{(i)})\mathrm{grad}[\delta p]\cdot\mathrm{grad}[q]\;\mathrm{d}\Omega. 
\end{align}
In each Newton iteration, we thus solve the
following linear variational problem:~Find
$\delta\mathbf{u} \in \mathcal{U}$ and
$\delta p \in \mathcal{Q}$ such that we
have
\begin{align}
\label{Eqn:NL_VMS}
&\mathcal{J}\left[\left(\mathbf{u}^{(i)},p^{(i)}\right);\left(\delta\mathbf{u},\delta p\right),\left(\mathbf{w},q\right)\right] = -\mathcal{F}\left[\left(\mathbf{u}^{(i)},p^{(i)}\right);(\mathbf{w},q)\right] = 0 \quad \forall \mathbf{w} \in \mathcal{W}, \;
\forall q \in \mathcal{Q}. 
\end{align}
We obtain the solution in an iterative fashion 
using the following update equation until the 
residual meets the prescribed tolerance:
\begin{subequations}
\label{Eqn:NL_update}
\begin{align}
&\mathbf{u}^{(i+1)} = \mathbf{u}^{(i)} + \delta\mathbf{u} \; \mathrm{and}\\
&p^{(i+1)} = p^{(i)} + \delta p.
\end{align}
\end{subequations}

\subsection{VI formulation in continuous setting}
In order to enforce the bound constraints due to 
maximum principles and the non-negative constraint 
we pose the problem as a variational inequality. 
To this end, we define the feasible solution 
space to be as follow: 
\begin{align}
  \mathcal{K} = \{(\mathbf{u},p) \in \mathcal{U}\times\mathcal{Q} \; 
  \vert \; p_{\mathrm{min}} \leq p \leq p_{\mathrm{max}}\}
\end{align}
In a specific problem, if there is no restriction on 
the lower bound of the pressure then one can set 
$p_{\mathrm{min}} = -\infty$. Similarly, one can set 
$p_{\mathrm{max}} = +\infty$ if there is no upper 
bound on the pressure.  
The proposed variational inequality in the 
continuous setting reads:~Find $(\mathbf{u},p) 
\in \mathcal{K}$ such that we have 
\begin{align}
  \mathcal{F}[\mathbf{u},p;\mathbf{w} 
    - \mathbf{u}, q - p] \geq 0 \qquad 
  \forall (\mathbf{w},q) \in \mathcal{K}
\end{align}

%% file: Sections/S4_NN_VI.tex
\section{PROPOSED COMPUTATIONAL FRAMEWORK}
\label{Sec:S4_NN_VI}
\subsection{Solver methodology}
Our proposed computational framework based on the 
two non-linear finite element variational formulations result 
in saddle-point problems, which are notoriously difficult to solve in a 
large-scale setting. Several classes of iterative solvers and preconditioning 
strategies exist for these types of problems \citep{benzi2005numerical,elman2006finite,murphy2000note}. 
One could alternatively employ hybridization techniques \citep{cockburn2009unified} 
which introduces Lagrange multipliers which can also significantly 
reduce the difficulty of solving such problems. 
However, in this study, we employ a Schur complement approach 
to precondition the saddle-point system. 

The residual vector $\boldsymbol{F}$ for
the RT0 formulation can be written as:
\begin{subequations}
  \begin{align}
    \label{Eqn:S4_residual_RT0}
    \boldsymbol{F}_{u} &:= \int_{\Omega} \alpha(p)
    \mathbf{u}\cdot\mathbf{w}\;\mathrm{d} \Omega 
    - \int_{\Omega} p\cdot\mathrm{div}[\mathbf{w}]\;\mathrm{d}\Omega
    + \int_{\Gamma^{\mathrm{P}}} p_0 \cdot \left(\mathbf{w}\cdot\mathbf{\hat{n}}\right)\; \mathrm{d} \Gamma\nonumber \\
     &\qquad- \int_{\Omega}\rho\mathbf{b}\cdot\mathbf{w}\;\mathrm{d} \Omega \; \mathrm{and} \\
\boldsymbol{F}_{p} &:= \int_{\Omega}q\cdot f\;\mathrm{d}\Omega 
    - \int_{\Omega}\mathrm{div}\left[\mathbf{u}\right]\cdot q\;\mathrm{d}\Omega, 
\end{align}
\end{subequations}
where the subscripts $u$ and $p$ denote
the velocity and pressure components
respectively. Likewise, the residual
vector for the VMS formulation is
written as follows:
\begin{subequations}
  \begin{align}
    \label{Eqn:S4_residual_VMS}
    \boldsymbol{F}_{u} &:= \frac{1}{2}\int_{\Omega} \alpha(p)\mathbf{u}\cdot\mathbf{w}\;\mathrm{d} \Omega 
    - \int_{\Omega} p\cdot\mathrm{div}[\mathbf{w}]\;\mathrm{d}\Omega
    + \int_{\Gamma^{\mathrm{P}}} p_0\cdot \left(\mathbf{w}\cdot\mathbf{\hat{n}}\right)\; \mathrm{d} \Gamma\nonumber\\
     &\qquad- \frac{1}{2}\int_{\Omega}\rho\mathbf{b}\cdot\mathbf{w}\;\mathrm{d} \Omega 
     -\frac{1}{2}\int_{\Omega}\mathbf{w}\cdot\mathrm{grad}[p]\;\mathrm{d}\Omega \; \mathrm{and} \\
\boldsymbol{F}_{p} &:= \int_{\Omega}q\cdot f\;\mathrm{d}\Omega 
    - \int_{\Omega}\mathrm{div}\left[\mathbf{u}\right]\cdot q\;\mathrm{d}\Omega
    - \frac{1}{2}\int_{\Omega}\mathrm{grad}[q]\cdot\mathbf{u}\;\mathrm{d}\Omega \nonumber\\
    &\qquad- \frac{1}{2}\int_{\Omega}\alpha^{-1}\mathrm{grad}[q]\cdot\left(\mathrm{grad}[p] 
    - \rho\mathbf{b}\right)\mathrm{d}\Omega.
\end{align}
\end{subequations}
The components of the Jacobian matrices for equations \eqref{Eqn:S4_residual_RT0} and
\eqref{Eqn:S4_residual_VMS}, respectively, can be subdivided as follows:
\begin{subequations}
\label{Eqn:S4_jacobian_RT0}
\begin{align}
\boldsymbol{J}_{uu} &:= \int_{\Omega} \alpha(p^{(i)})\delta\mathbf{u}\cdot\mathbf{w}\;\mathrm{d} \Omega, \\
\boldsymbol{J}_{up} &:= \int_{\Omega} \frac{\partial \alpha(p^{(i)})}{\partial p}\mathbf{u}^{(i)}\delta p\cdot\mathbf{w}\;\mathrm{d} \Omega - \int_{\Omega} \delta p\cdot\mathrm{div}[\mathbf{w}]\;\mathrm{d}\Omega, \\
\boldsymbol{J}_{pu} &:= -\int_{\Omega} \mathrm{div}\left[\delta\mathbf{u}\right]\cdot q\;\mathrm{d}\Omega, \; \mathrm{and} \\
\boldsymbol{J}_{pp} &:= \boldsymbol{0},
\end{align}
\end{subequations}
and
\begin{subequations}
\label{Eqn:S4_jacobian_VMS}
\begin{align}
\boldsymbol{J}_{uu} &:= \frac{1}{2}\int_{\Omega} \alpha(p^{(i)})\delta\mathbf{u}\cdot\mathbf{w}\;\mathrm{d} \Omega, \\
\boldsymbol{J}_{up} &:= \frac{1}{2}\int_{\Omega} \frac{\partial \alpha(p^{(i)})}{\partial p}\mathbf{u}^{(i)}\delta p\cdot\mathbf{w}\;\mathrm{d} \Omega - \int_{\Omega} \delta p\cdot\mathrm{div}[\mathbf{w}]\;\mathrm{d}\Omega - \frac{1}{2}\int_{\Omega}\mathrm{grad}[\delta p]\cdot\mathbf{w}\;\mathrm{d}\Omega, \\
\boldsymbol{J}_{pu} &:= -\int_{\Omega} \mathrm{div}\left[\delta\mathbf{u}\right]\cdot q\;\mathrm{d}\Omega 
- \frac{1}{2}\int_{\Omega}\delta\mathbf{u}\cdot\mathrm{grad}[q]\;\mathrm{d}\Omega, \; \mathrm{and} \\
\boldsymbol{J}_{pp} &:= - \frac{1}{2}\int_{\Omega}\alpha^{-2}(p^{(i)})\frac{\partial \alpha(p^{(i)})}{\partial p}\left(\mathrm{grad}\left[p^{(i)}\right]-\rho\mathbf{b}\right)\delta p\cdot\mathrm{grad}[q]\;\mathrm{d}\Omega\nonumber\\
&\qquad- \frac{1}{2}\int_{\Omega}\alpha^{-1}(p^{(i)})\mathrm{grad}[\delta p]\cdot\mathrm{grad}[q]\;\mathrm{d}\Omega.
\end{align}
\end{subequations}
Conceptually, the problem 
at hand is a 2$\times$2 block 
matrix:
\begin{align}
\boldsymbol{J} = \begin{pmatrix}
\boldsymbol{J}_{uu} & \boldsymbol{J}_{up}\\
\boldsymbol{J}_{pu} & \boldsymbol{J}_{pp}
\end{pmatrix},
\end{align}
which admits a full factorization of
\begin{align}
\boldsymbol{J} = \begin{pmatrix}
\boldsymbol{I} & \boldsymbol{0}\\
\boldsymbol{J}_{pu}\boldsymbol{J}_{uu}^{-1} & \boldsymbol{I}
\end{pmatrix}
\begin{pmatrix}
\boldsymbol{J}_{uu} & \boldsymbol{0}\\
\boldsymbol{0} & \boldsymbol{S}
\end{pmatrix}
\begin{pmatrix}
\boldsymbol{I} & \boldsymbol{J}_{uu}^{-1}\boldsymbol{J}_{up}\\
\boldsymbol{0} & \boldsymbol{I}
\end{pmatrix},
\end{align}
where $\boldsymbol{I}$ is the identity matrix and
\begin{align}
\boldsymbol{S}=\boldsymbol{J}_{pp}-\boldsymbol{J}_{pu}\boldsymbol{J}_{uu}^{-1}\boldsymbol{J}_{up},
\end{align}
is the Schur complement. The inverse can therefore be written as
\begin{align}
\boldsymbol{J}^{-1} = \begin{pmatrix}
\boldsymbol{I} & -\boldsymbol{J}_{uu}^{-1}\boldsymbol{J}_{up}\\
\boldsymbol{0} & \boldsymbol{I}
\end{pmatrix}
\begin{pmatrix}
\boldsymbol{J}_{uu}^{-1} & \boldsymbol{0}\\
\boldsymbol{0} & \boldsymbol{S}^{-1}
\end{pmatrix}
\begin{pmatrix}
\boldsymbol{I} & \boldsymbol{0}\\
-\boldsymbol{J}_{pu}\boldsymbol{J}_{uu}^{-1} & \boldsymbol{I}
\end{pmatrix}.
\end{align}
The task at hand is to approximate $\boldsymbol{J}_{uu}^{-1}$ and
$\boldsymbol{S}^{-1}$. Since $\boldsymbol{J}_{uu}$ is a mass matrix
for the Darcy equation, we can invert it using the
ILU(0) (incomplete lower upper) solver. We employ
a diagonal mass-lumping of $\boldsymbol{J}_{uu}$ to 
estimate $\boldsymbol{J}_{uu}^{-1}$. That is,
\begin{align}
\boldsymbol{S}_p = \boldsymbol{J}_{pp}-\boldsymbol{J}_{pu}\mathrm{diag}\left(
\boldsymbol{J}_{uu}\right)^{-1}\boldsymbol{J}_{up},
\end{align}
to precondition the inner solver inverting $\boldsymbol{S}$.
For this block we employ the multi-grid V-cycle on $\boldsymbol{S}_p$
from the HYPRE BoomerAMG package (\citep{hypre-users-manual}).
As discussed in Appendix B of \citep{chang2017variational}, the
$\boldsymbol{J}_{uu}^{-1}$ and $\boldsymbol{S}_p$,
only a single sweep of ILU(0) and HYPRE's V-cycle is needed
for the $\boldsymbol{J}_{uu}^{-1}$ and $\boldsymbol{S}_p$ matrices, 
respectively, and the GMRES method is employed to solve the entire block system.
\subsection{Variational inequality approach}
We denote the total number of degrees-of-freedom by ``$ndofs$''. 
The component-wise inequalities are denoted
by $\preceq$ and $\succeq$. That is,
\begin{subequations}
  \begin{align}
    &\boldsymbol{a} \preceq \boldsymbol{b}
    \quad \mbox{implies that} \quad a_{n}
    \leq b_n \; \forall \; n  \; \mathrm{and} \\
    &\boldsymbol{a} \succeq \boldsymbol{b}
    \quad \mbox{implies that} \quad a_{n}
    \geq b_n \; \forall \; n.
  \end{align}
\end{subequations}
The standard inner-product in Euclidean spaces is
denoted by $\langle \cdot ; \cdot \rangle$. That is,
\begin{align}
  \langle \boldsymbol{a} ; \boldsymbol{b} \rangle
  = \sum_{n}^{ndofs} a_n b_n \quad \forall \boldsymbol{a},
  \boldsymbol{b} \in \mathbb{R}^{ndofs}.
\end{align}
Let $\boldsymbol{u}$ and $\boldsymbol{p}$ denote the discrete vector of unknowns for 
velocity and pressure respectively. The vector of all degrees-of-freedom, denoted by
$\boldsymbol{v}\in\mathbb{R}^{ndofs}$, can be defined as:
\begin{align}
& \boldsymbol{v} := \left\{
		\begin{array}{l}
		\boldsymbol{u}\\ 
		\boldsymbol{p}
		\end{array} 
		\right\}.
\end{align}
For convenience, let us also define the following functional
$\boldsymbol{F}(\boldsymbol{v})\in\mathbb{R}^{ndofs}$ as
\begin{align}
&\boldsymbol{F}(\boldsymbol{v}) := \left\{
		\begin{array}{l}
		\boldsymbol{F}_{u}\\ 
		\boldsymbol{F}_{p}
		\end{array} 
		\right\}.
\end{align}
The VI formulation in the discrete setting is
posed as a \emph{Mixed Complementarity Problem
(MCP)} \citep{Kinderlehrer_VI_2000}:~Find 
$\boldsymbol{v}_{\mathrm{min}} \preceq
\boldsymbol{v} \preceq \boldsymbol{v}_{\mathrm{max}}$
such that for each $n\in\left\{1,...,ndofs\right\}$
\begin{subequations}
  \begin{alignat}{2}
    &F_{n}(\boldsymbol{v}) \geq 0 \quad
    &&\mathrm{if} \; v_{\mathrm{min}} = v_{n},  \\
    &F_{n}(\boldsymbol{v}) = 0 \quad
    &&\mathrm{if} \; v_{\mathrm{min}} \leq v_{n} \leq v_{\mathrm{max}}, \; \mathrm{and} \\
    &F_{n}(\boldsymbol{v}) \leq 0 \quad
    &&\mathrm{if} \; v_{n} =v_{\mathrm{max}}.
  \end{alignat}
\end{subequations}
where $v_{\mathrm{min}}$ and $v_{\mathrm{max}}$, respectively,
denote the minimum and maximum values for pressure and velocity. 
The constraints for pressure ($p_{\mathrm{min}}$ and $p_{\mathrm{max}}$) 
are provided by the maximum principle, whereas the minimum and 
maximum constraints for each directional component of velocity 
are $-\infty$ and $+\infty$ respectively.

If one has only lower bound constraints due to 
the presence of a positive pressure source 
(i.e., $f > 0$, $p_{\mathrm{min}} = 0$, and 
$p_{\mathrm{max}} = +\infty$), then the VI reduces 
to a non-linear complementarity problem, 
which is a special case of MCP. For details 
on non-linear complementarity problems, 
see \citep{Facchinei_FDVI_2003}.
Note that the feasible region, which is restricted by
the bound constraints, form a parallelepiped, which
is a convex set \citep{Boyd_CO_2004}. 

Let the feasible region $\mathcal{K}$ be a convex subset
of $\mathbb{R}^{ndofs}$. In our case, the feasible region
is restricted by constraints which are in the form of
finite number of linear equalities and inequalities. This
makes the feasible region to be a polyhedron, which is a
convex set \citep{Boyd_CO_2004}. It should be
noted that bound constraints are a special case of linear
inequalities. With this machinery at our disposal, one can
pose the second formulation based on variational inequalities,
which reads:~Find $\boldsymbol{v} \in \mathcal{K}$ such that
we have
\begin{align}
  \label{Eqn:S4_vi_formulation}
  \langle \boldsymbol{F}(\boldsymbol{v});\boldsymbol{w} - \boldsymbol{v}\rangle
  \geq \boldsymbol{0}
  \quad \forall \boldsymbol{w} \in \mathcal{K}.
\end{align}
\subsection{Computer implementation}
In this paper, we implement the proposed variational inequality
based computational framework using the Firedrake Project 
\citep{Rathgeber_ACM_2015,Luporini_ACM_2016,
Luporini_ACMACO_2015}. It is a python-based library that provides an automated 
system for the solution of partial differential equations using the finite 
element method. The MPI-based PETSc library is utilized as the parallel linear
algebra back-end. These solvers have been demonstrated to show good parallel
scalability for large-scale optimization-based problems \citep{Chang_JOSC_2017}.

The PETSc library \citep{balay2014petsc} provides a wide array of 
solvers for finite-dimensional 
VI's. For example, two popular algorithms are the 
\emph{semi-smooth Newton (SS)} \citep{DeLuca_MP_1996,Munson_INFORMS_2001} and
\emph{Reduced-space active-set (RS)} \citep{Benson_OMS_2006} methods. It has been shown
in \citep{Benson_OMS_2006} that the performance of SS and RS methods are 
application dependent and in \citep{chang2017variational} that the RS 
method demonstrates better solver convergence for advection-diffusion type 
equations. However, preliminary results (not shown in the paper) suggest that
it is in fact the SS method that performs better for the nonlinear flow model. Thus, we propose
the following algorithm:
\begin{enumerate}
\item Read in mesh, boundary conditions, and material properties.
\item Solve for $\boldsymbol{v}$ with no constraints (call it $\boldsymbol{v}_0$).
\item CONDITIONAL: If $\boldsymbol{p}$ violates the discrete maximum principles:
\begin {enumerate}
\item Set $\boldsymbol{v}_0$ as initial guess for SS method.
\item Solve for $\boldsymbol{v}$ using SS method.
\end{enumerate}
\end{enumerate}

In the next section, all 2D problems will be conducted in serial on an 
Intel Xeon E5-2609v3 (Haswell) processor, and the 3D problem will be 
conducted in parallel on an Intel Xeon Phi 7250 (Knights Landing) processor.

%% file: Sections/S5_NN_NR.tex
`\section{REPRESENTATIVE NUMERICAL RESULTS}
\label{Sec:S5_NN_NR} 
\subsection{$h$ - convergence study}
\label{sec:h_conv}
\begin{figure}[h!]
  \centering
  \subfigure{\includegraphics[scale=0.85]{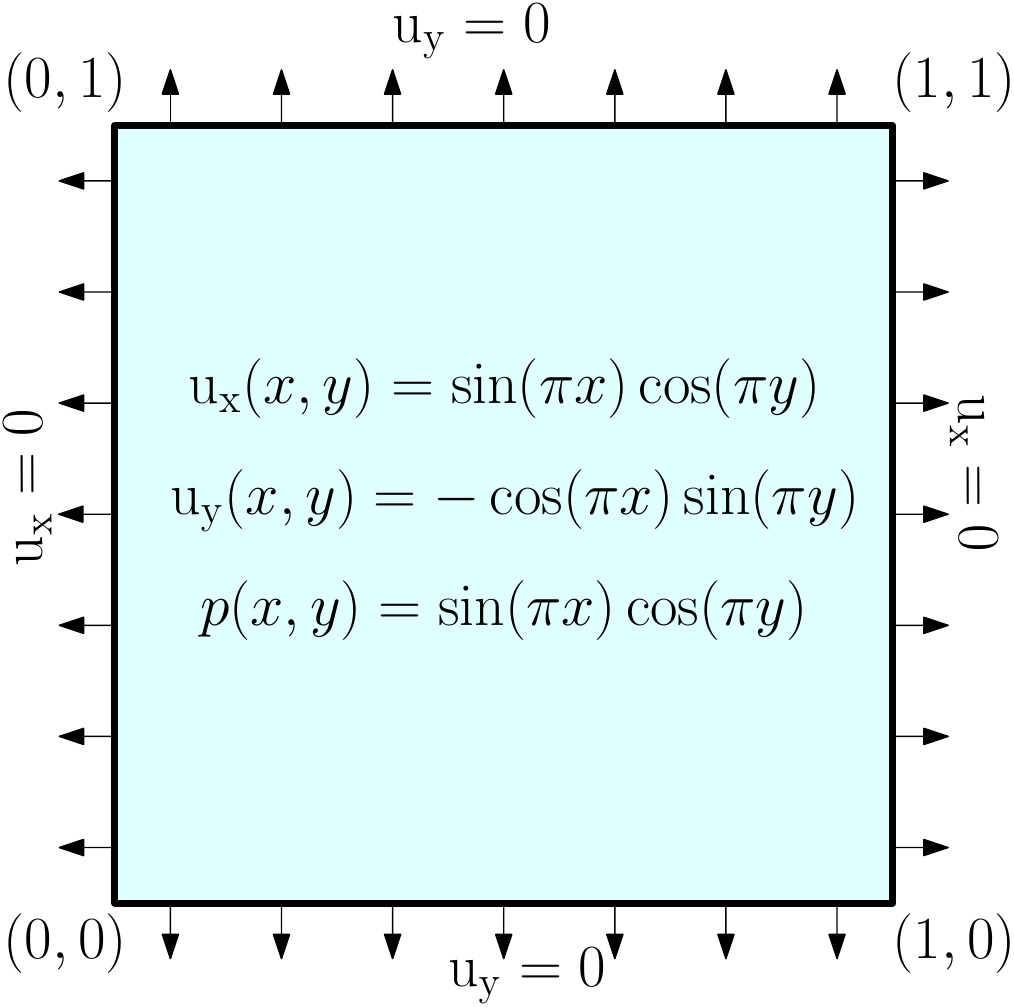}}
  \captionsetup{format=hang}
  \caption{Problem statement for the $h$-convergence study, 
    showing the computational domain, pressure and velocity 
    functions and the boundary conditions.}
  \label{fig:h_conv_q}
\end{figure}
We first perform an $h$-convergence study on the mixed formulations 
to verify that the Firedrake project library and proposed solver 
methodologies are converging schemes. A finite element solution 
is said to be converging if the difference between the exact and 
numerical solutions decreases with the mesh refinement. Consider 
a unit square to be the computational domain with the following 
expressions for the pressure and velocity fields:
\begin{subequations}\label{eq:h-conv}
	\begin{align}
		& \mathbf{u}(x,y,z) = 
		\left\{
		\begin{array}{ll}
		\sin(\pi x)\cos(\pi y)\\ -\cos(\pi x)\sin(\pi y)
		\end{array}\right\}
		\; \mathrm{and}\\
		& p(x,y) = \sin(\pi x)\sin(\pi y).			
	\end{align}
\end{subequations}
Through the method of manufactured solutions, by substituting equation 
\eqref{eq:h-conv} into equation \eqref{eq:modified_darcy} we obtain 
the following expression for the body force:
\begin{equation}
		 \mathbf{b}(x,y)  = \frac{1}{\rho}
		\left[\mathbf{\alpha}(p(x,y))\right]^{-1}\left\{
		                \begin{array}{l} \sin(\pi x)\cos(\pi y)+ \sin^{2}(\pi x) \sin(\pi y) \cos(\pi y) \\ 
		                  \qquad+ \pi \cos(\pi x) \sin(\pi y) \\ 
		                  -\cos(\pi x)\sin(\pi y) - \sin(\pi x)\sin^{2}(\pi y)\cos(\pi x) \\ 
		                  \qquad+\pi \sin(\pi x)\cos(\pi y)
		                \end{array}
		                \right\},		
\end{equation}
where $\mathbf{\alpha}$ is given by equation 
\eqref{eq:barus_lin}. The boundary conditions for this problem are:
\begin{table}[t]
\centering
\captionsetup{format=hang}
\caption{User defined parameters for $h$-convergence study}
\begin{tabular}{c c}
\hline\hline
Parameter & Values \\
\hline
$\mu_0$ & 1 \\
$\beta_B$ & 0 and 1\\
$\rho$ & 1\\
$\mathbf{K}$ &  $\mathbf{I}$ \\ 
$\rho\mathbf{b}$ &  $\mathbf{0}$ \\ 
\hline
\label{table: h-conv}
\end{tabular}
\end{table}
%
%
\begin{equation}
  u_x(x =0, y) = u_x(x = 1, y) = 
  u_y(x, y = 0) = u_y(x, y =1) = 0.					
\end{equation}
Figure \ref{fig:h_conv_q} provides a pictorial description
of the problem, and Table \ref{table: h-conv} lists the
parameters employed in the numerical simulation. 
%
\begin{figure}[hbt]
	\hspace{1pt}
	\centering
	\subfigure{\includegraphics[scale=0.6]{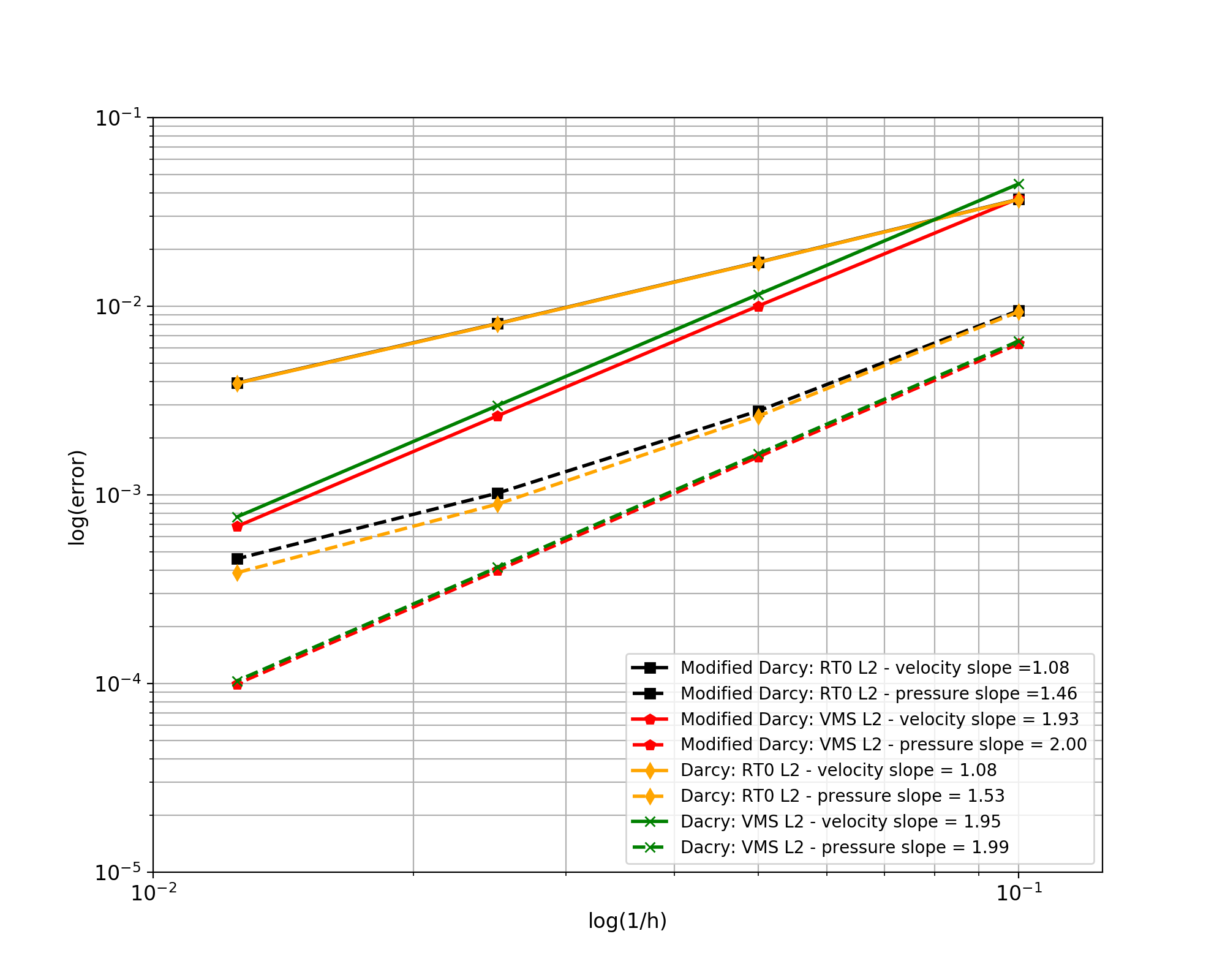}}
	\captionsetup{format=hang}
	\caption{Convergence plots of $L_2$ errors for RT0 and VMS formulations for the standard Darcy $(\beta_B = 0)$ and 
	modified Darcy $(\beta_B = 1)$ models.}
	\label{fig:h_conv}
\end{figure}
Figure \ref{fig:h_conv} provides a comparison between $L_2$
error norm convergence rates for RT0 and VMS. Theoretical
convergence rates of both velocity and pressure under RT0
is unity whereas the convergence rates for these two fields
under VMS is two. We see that the slopes for the standard
Darcy model (where $\beta_B = 0$) are similar to the
theoretical slopes, which verifies
the convergence of the Firedrake 
project's finite element framework. Furthermore, extending the 
mixed formulations to the modified Darcy model with pressure-dependent
viscosity (where $\beta_B > 0$) has similar convergence. The studies 
performed so far suggest that the Firedrake project is a suitable 
software package for conducting finite element simulations, and 
we now examine scenarios where VI is needed
to enforce maximum principles. 
\FloatBarrier
\subsection{Square reservoir}
\begin{figure}[h!]
	\centering
	\subfigure{\includegraphics[scale=0.71]{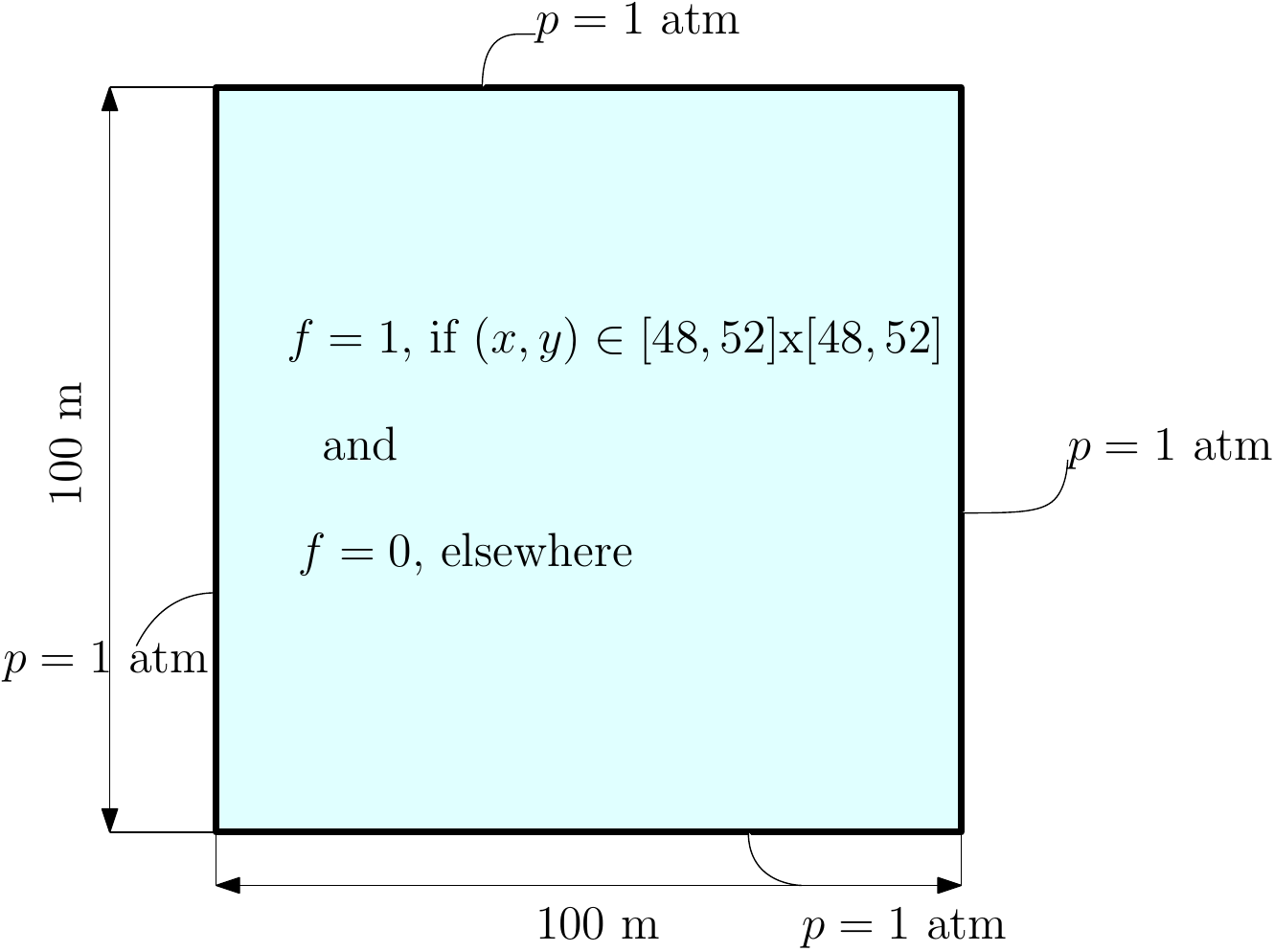}}
	\captionsetup{format=hang}
	\caption{Pictorial description of square reservoir problem. The computational 
	domain, boundary conditions, and volumetric source are shown.}
	\label{fig:bench_q}
\end{figure}
This 2D heterogeneous problem aims to 
illustrate not only the effectiveness of the proposed VI based framework 
to ensure DMP for the pressure field but also how levels of heterogeneity and anisotropy
affect the overall computational effort. Consider a square reservoir on a domain 
$\Omega$: = (0 m, 100 m) $\times$ (0 m, 100 m) with the following anisotropic 
heterogeneous permeability tensor:
\begin{align}
\mathbf{K} = k_0\begin{pmatrix}
y^2 + \epsilon x^2 & -(1-\epsilon)xy\\
-(1-\epsilon)xy & x^2 + \epsilon y^2
\end{pmatrix}\;\mathrm{m}^{-2},
\end{align}
where $k_0=10^{-13}$ m$^2$ denotes the base permeability and 
$\epsilon$ is a user-defined value that controls the level of anisotropy. 
Three values of $\epsilon$ are considered as shown in Table \ref{table: eps_cases}.
 \begin{table}[t]
 \centering
 \captionsetup{format=hang}
 \caption{Square reservoir: different values of $\epsilon$ used for this study.}
 \begin{tabular}{c c }
 \hline\hline
 Case ID & Value of $\epsilon$   \\
 \hline
 1 & $10^{-3}$  \\
 2 & $10^{-2}$ \\
 3 & $10^{-1}$ \\ 
 \hline
 \label{table: eps_cases}
 \end{tabular}
 \end{table}
 \begin{figure}[t!]
	\centering
	\subfigure[No VI ($\epsilon=10^{-3}$)]{\includegraphics[scale=0.4]{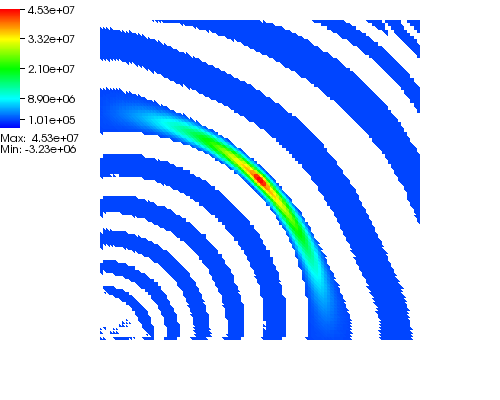}}	
	\subfigure[With VI ($\epsilon=10^{-3}$)]{\includegraphics[scale=0.4]{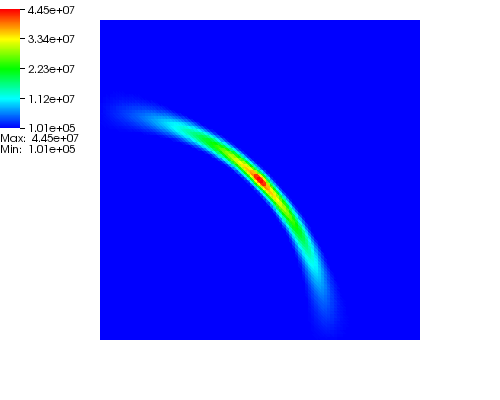}}\\
	\subfigure[No VI ($\epsilon=10^{-2}$)]{\includegraphics[scale=0.4]{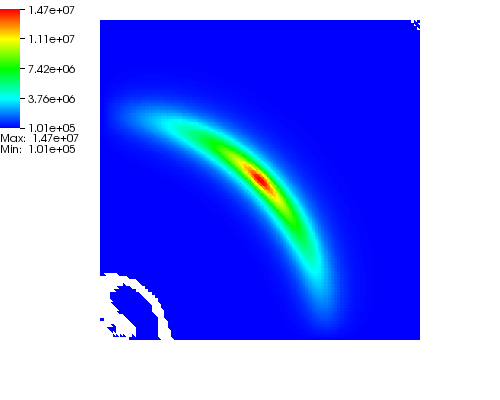}}
	\subfigure[With VI ($\epsilon=10^{-2}$)]{\includegraphics[scale=0.4]{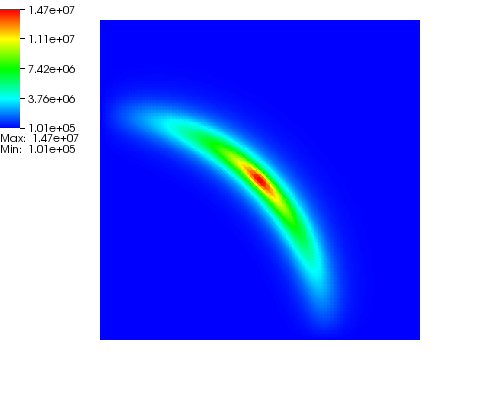}}\\
	\subfigure[No VI ($\epsilon=10^{-1}$)]{\includegraphics[scale=0.4]{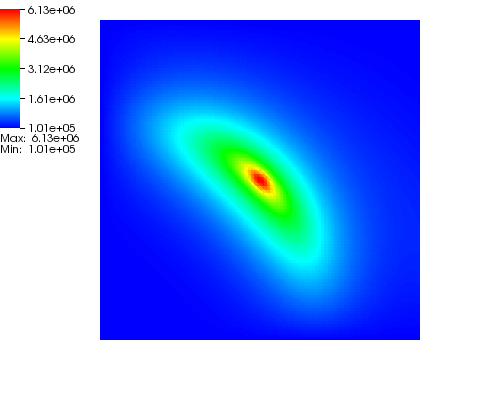}}
	\subfigure[With VI ($\epsilon=10^{-1}$)]{\includegraphics[scale=0.4]{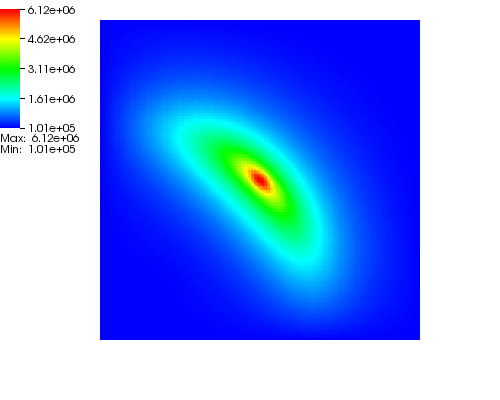}}	
\captionsetup{format=hang}
\caption{Square reservoir: pressure contours for RT0 formulation before (left) and after (right) VI. The white spaces are representative of DMP violations.}
\label{fig:pres_rt_vio}
\end{figure}
\begin{figure}[t!]
	\centering
	\subfigure[No VI ($\epsilon=10^{-3}$)]{\includegraphics[scale=0.4]{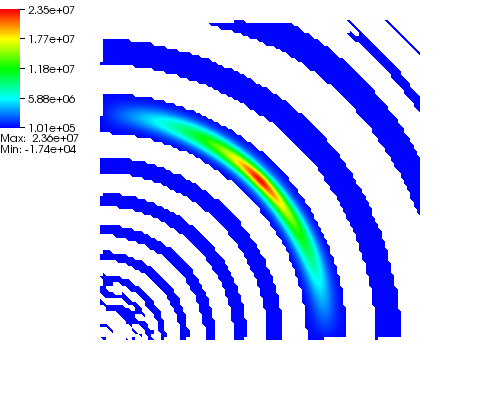}}	
	\subfigure[With VI ($\epsilon=10^{-3}$)]{\includegraphics[scale=0.4]{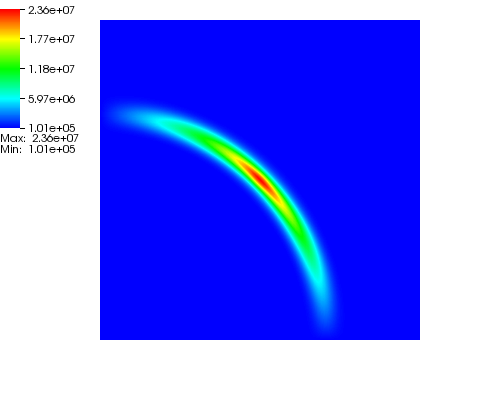}}\\
	\subfigure[No VI ($\epsilon=10^{-2}$)]{\includegraphics[scale=0.4]{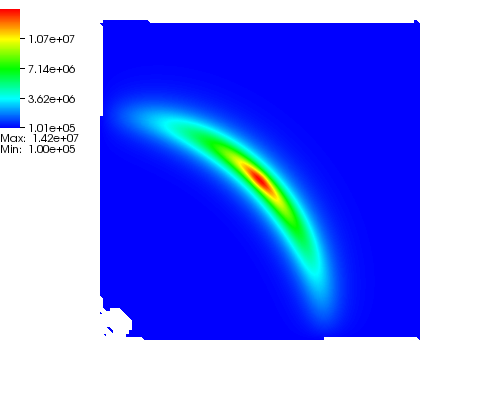}}
	\subfigure[With VI ($\epsilon=10^{-2}$)]{\includegraphics[scale=0.4]{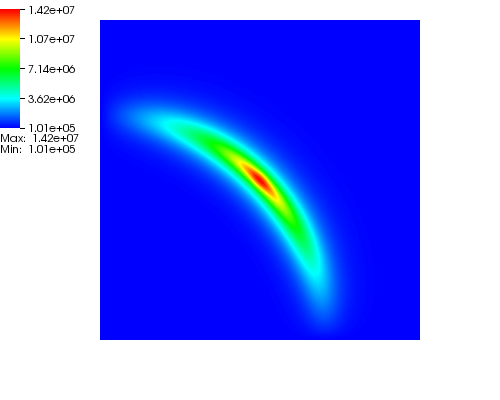}}\\
	\subfigure[No VI ($\epsilon=10^{-1}$)]{\includegraphics[scale=0.4]{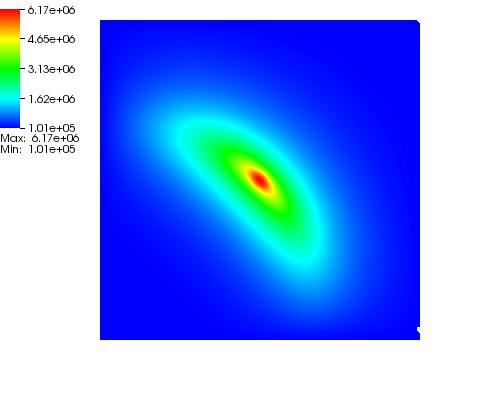}}
	\subfigure[With VI ($\epsilon=10^{-1}$)]{\includegraphics[scale=0.4]{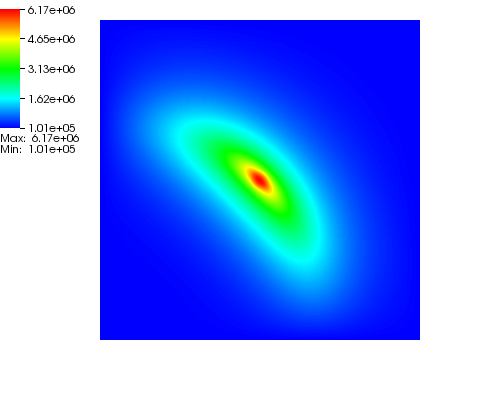}}	
\captionsetup{format=hang}
\caption{Square reservoir: pressure contours for VMS formulation before (left) and 
after (right) VI. The white spaces are representative of DMP violations.}
\label{fig:pres_vms_vio}
\end{figure}
We assume that $\mu=10^{-3}$ Pa$\cdot$s and $\beta_B = 10^{-8}$ Pa$^{-1}$. 
A constant pressure of 101325 Pa (1 atm) is applied on the entire boundary for 
both mixed formulations, see Figure \ref{fig:bench_q} for a pictorial description. 
Let $\rho\mathbf{b} = \mathbf{0}$ and the volumetric source $f$ be given by:
   \begin{equation}
     f=
     \begin{cases}
       1\;\mathrm{s}^{-1}, & \text{if}\ (x,y) \in [48,52] \mathrm{x} [48,52]\\
       0, & \text{elsewhere}
     \end{cases}.
   \end{equation}
Since a positive forcing function is present, only the lower bound constraint
is enforced in the VI framework. Both RT0 and VMS formulations are employed 
for the numerical discretization, 
with triangular elements of $\mathit{h}$-size = 1 m. We perform 
this study for three  
cases of $ \epsilon$ as listed in Table \ref{table: eps_cases}.
  
Figures \ref{fig:pres_rt_vio} and \ref{fig:pres_vms_vio} show the pressure 
contours of the RT0 and VMS formulations, respectively, for different values of 
$\epsilon$. The white regions are representative of DMP violations of pressure. 
Moreover, the effect of enforcing the VI framework over the RT0 and VMS formulations 
is also shown. These figures demonstrate that the VI framework is capable of enforcing 
the lower bound constraints for pressure values. It is also interesting to note that 
a smaller $\epsilon$ results in more violations regardless of the finite
element discretization. Tables \ref{table:Benchmark_RT0_Iter} and 
\ref{table:Benchmark_VMS_Iter} illustrate
the effect of $\epsilon$ on both the number of violating cells as well as solver
performance. Small values of $\epsilon$ make the systems of equation much harder to
solve as both time to solution and number of solver iterations increase. However, 
it can be seen that the additional computational cost from the VI solver is 
not significant; it was shown in \citep{chang2017variational} that the same 
computational framework with the RS method for
advection-diffusion equations increased time to solution 
by up to a factor of 20 whereas for this particular problem the SS 
method increased time to solution by no more than 66 percent. 

In a mixed formulation, altering the pressures
or velocities may have a direct impact over
their counterparts, but such changes to the
velocity imposed by the VI framework do not
have a negative impact on the overall numerical
accuracy. The velocities under the VI framework
for the RT0 and VMS formulations are shown in
Figures \ref{fig:vel_rt} and \ref{fig:vel_vms},
respectively. Absolute differences between
velocity fields obtained from the mixed
formulations without the imposed VI framework
and the velocity fields obtained from the mixed
formulations under the imposed VI framework are
shown in those figures. These figures suggest
that enforcement of the DMP on pressure does
impact the velocities. Some formulations like
RT0 have larger differences in velocity magnitudes
whereas for others like VMS, it may not be as large. 
However, any such change to the velocity field can
have a significant impact on subsurface transport
especially for long periods of time.
\begin{table}[t!]
  \captionsetup{format=hang}
  \caption{Square reservoir: computational
    results and initial violations for RT0.}
    \begin{tabular}{ c| c c c| c c c | c |c }
      \hline
      \hline
    
       \multirow{1}{*}{Case} &
         \multicolumn{3}{c}{RT0} &
         \multicolumn{3}{c}{VI over RT0} &
       \multirow{1}{*}{Total} &
       \multirow{1}{*}{ \% Vio-} \\
         
      ID & KSP & SNES & Time & KSP & SNES & Time & Time & lations \\
 \hline      
  1 & 2059 &	8 &	1.98E+001 &	664	& 3	& 4.53E+000	& 2.43E+001	& 50.54	\\
  2 & 227  & 4 &	2.63E+000 &	125	& 2	& 1.54E+000	& 4.18E+000	& 2.06 \\
  3 & 53	& 3 & 9.53E-001	& 36 & 2 & 6.26E-001 & 1.58E+000 & 0.02 \\
       \hline    
    \end{tabular}
    \label{table:Benchmark_RT0_Iter}
  \end{table}  
  \begin{table}[t!]
  \captionsetup{format=hang}
  \centering
  \caption{Square reservoir: computational results and initial violations for VMS.}
    \begin{tabular}{ c |c c c| c c c | c |c }
      \hline
      \hline   
       \multirow{1}{*}{Case} &
         \multicolumn{3}{c}{VMS} &
         \multicolumn{3}{c}{VI over VMS} &
       \multirow{1}{*}{Total} &
       \multirow{1}{*}{ \% Vio-} \\
         
      ID & KSP & SNES & Time & KSP & SNES & Time & Time & lations \\
      \hline        
  1 & 608 & 5 & 1.75E+001 & 47 & 2 & 3.08E+000 & 2.06E+001 &	55.67 \\
  2 & 185 & 4 & 8.31E+000 & 71 &	2 &	3.80E+000 &	1.21E+001 &	2.37 \\
  3 & 51	& 3	& 5.00E+000	& 32 & 2 & 3.29E+000 & 8.29E+000 & 1.44	\\
       \hline    
    \end{tabular}
    \label{table:Benchmark_VMS_Iter}
  \end{table}  
 \begin{figure}[t]
	\centering
	\subfigure[VI velocity ($\epsilon=10^{-3}$)]{\includegraphics[scale=0.4]{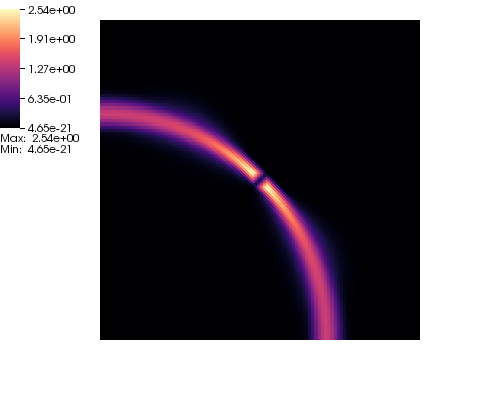}}
	\subfigure[Absolute diff. ($\epsilon=10^{-3}$)]{\includegraphics[scale=0.4]{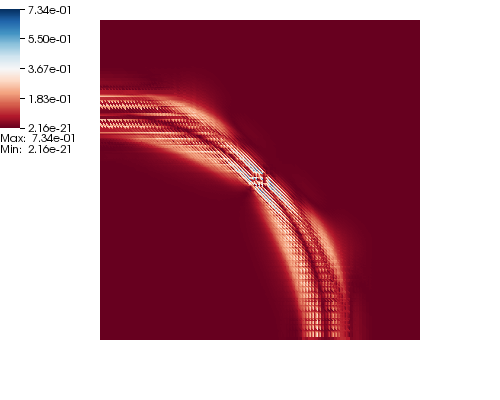}}\\
	\subfigure[VI velocity ($\epsilon=10^{-2}$)]{\includegraphics[scale=0.4]{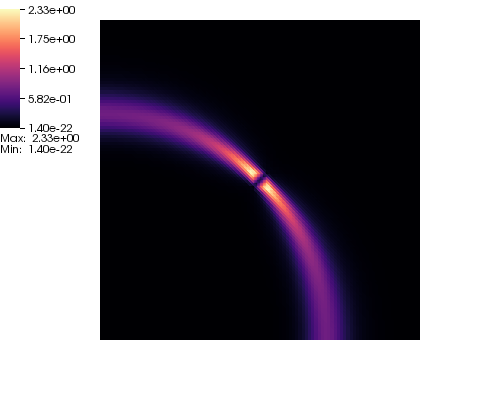}}
	\subfigure[Absolute diff. ($\epsilon=10^{-2}$)]{\includegraphics[scale=0.4]{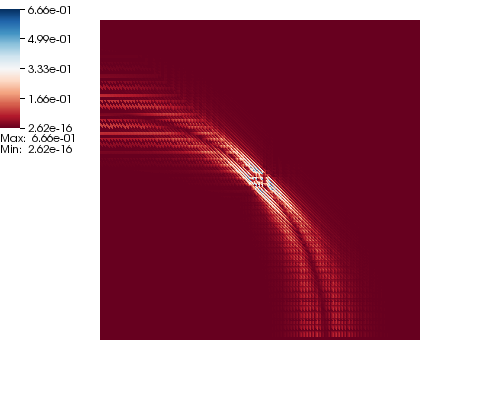}}\\
	\subfigure[VI velocity ($\epsilon=10^{-1}$)]{\includegraphics[scale=0.4]{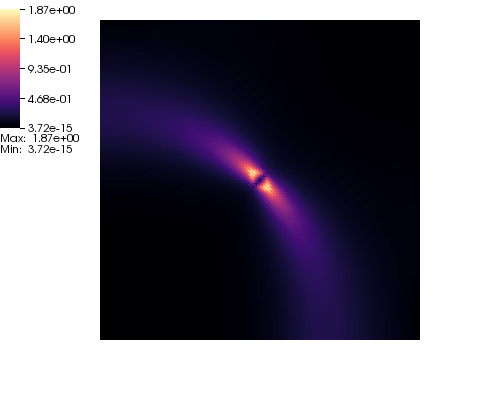}}	
	\subfigure[Absolute diff. ($\epsilon=10^{-1}$)]{\includegraphics[scale=0.4]{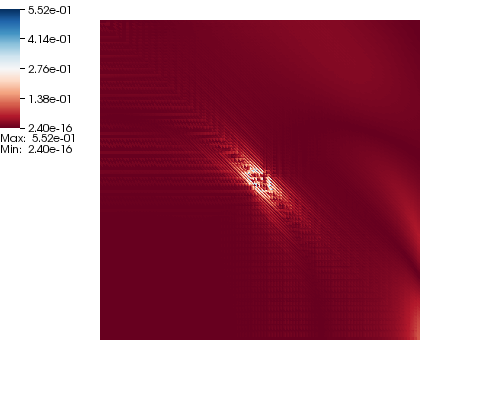}}
\captionsetup{format=hang}
\caption{Square reservoir: velocity profiles for VI based framework imposed over RT0 formulation (left) and the absolute differences between the non-VI and VI velocities (right).}
\label{fig:vel_rt}
\end{figure}
\begin{figure}[t]
	\centering
	\subfigure[VI velocity ($\epsilon=10^{-3}$)]{\includegraphics[scale=0.4]{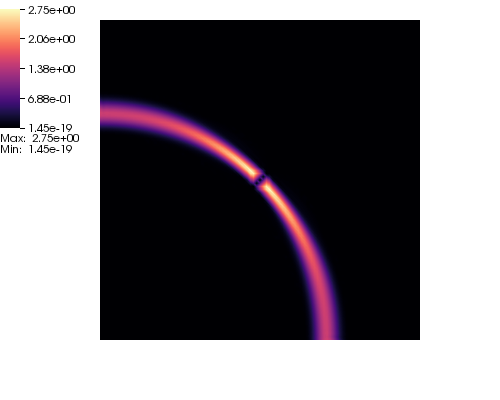}}
	\subfigure[Absolute diff. ($\epsilon=10^{-3}$)]{\includegraphics[scale=0.4]{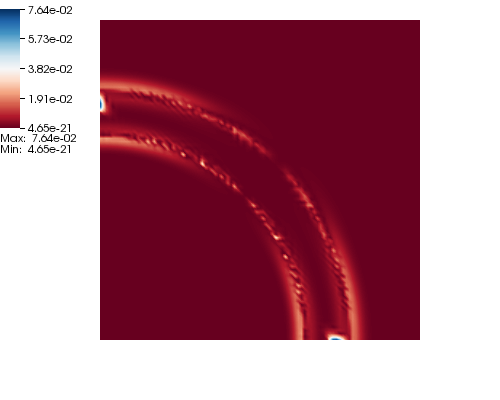}}\\
	\subfigure[VI velocity ($\epsilon=10^{-2}$)]{\includegraphics[scale=0.4]{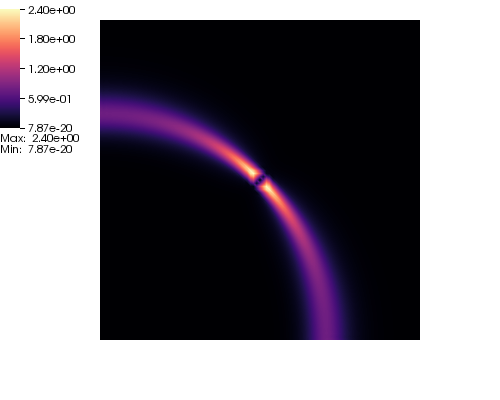}}	
	\subfigure[Absolute diff. ($\epsilon=10^{-2}$)]{\includegraphics[scale=0.4]{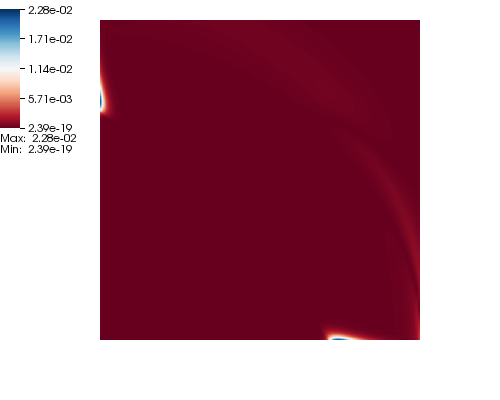}}\\
	\subfigure[VI velocity ($\epsilon=10^{-1}$)]{\includegraphics[scale=0.4]{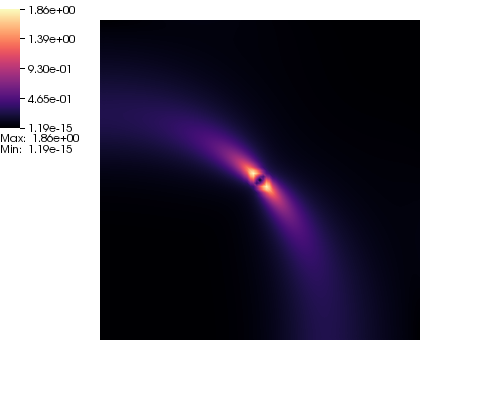}}	
	\subfigure[Absolute diff. ($\epsilon=10^{-1}$)]{\includegraphics[scale=0.4]{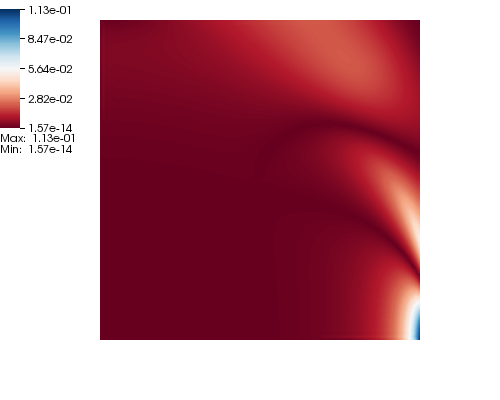}}
\captionsetup{format=hang}
\caption{Square reservoir: velocity profiles for VI based framework imposed over VMS formulation (left) and the absolute differences between the non-VI and VI velocities (right).}
\label{fig:vel_vms}
\end{figure}
\FloatBarrier
\subsection{Circular reservoir}
\label{sec:s_s}
\begin{figure}[hbt]
  \centering
  \captionsetup{format=hang}
  \subfigure{\includegraphics[scale =0.61]{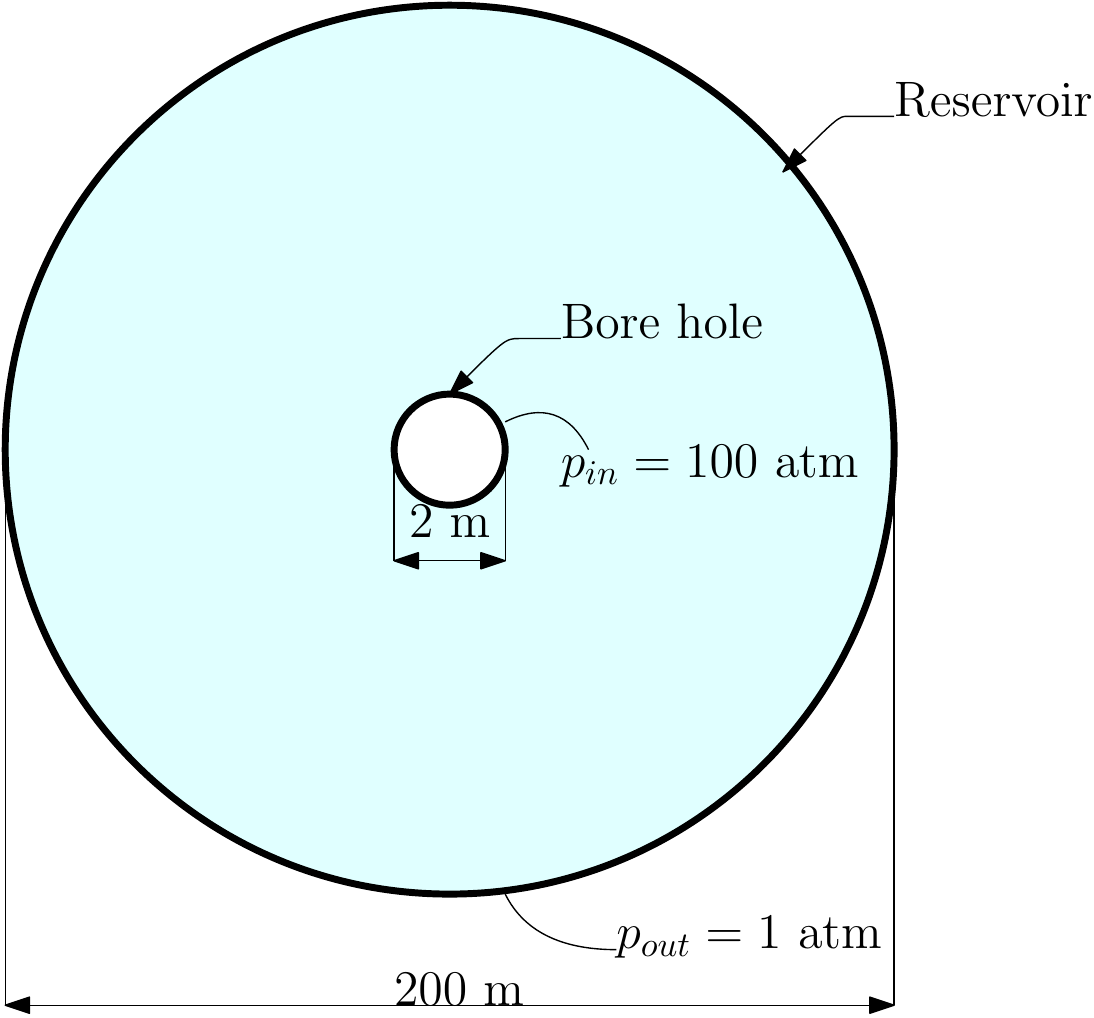}}
  \caption{Pictorial description of the circular 
    reservoir problem, showing the computational 
    domain and boundary conditions.}
  \label{fig:res_q}
\end{figure}

For this problem, the task is to study the algorithmic scalability of the 
proposed computational framework when different unstructured grids and different
$\beta_B$ coefficients are employed. Consider the circular reservoir 
shown in Figure \ref{fig:res_q} which has a 100 m outer radius and an inner circular 
borehole radius of 1 m. Neither the specific body force nor the volumetric force is
present for this problem so the VI framework imposes both lower
bound and upper bound constraints. The lower bound or outer boundary is maintained 
at the atmospheric pressure ($p_{out} = 1.0\times 10^{5}$ Pa). At the injection hole, 
a constant pressure of 100 atm is maintained ($p_{in} = 1.0 \times 10^{7}$ Pa) and 
serves as the upper bound constraint. The different meshes used for this study, as well 
as the corresponding numbers of degrees-of-freedom for the RT0 and VMS formulations, 
are provided in Table \ref{table:mesh_id}. Each mesh ID and mixed formulation 
shall be simulated with various values of $\beta_B$ provided in Table \ref{table:beta_case}.
\begin{table}[t!]
\centering
\captionsetup{format=hang}
\caption{Hierarchy of meshes for the circular reservoir problem.}
  \begin{tabular}{c c c| c c c |c c c }
    \hline
    \hline
    \multirow{1}{*}{Mesh} &
    \multirow{1}{*}{No.of} &
    \multirow{1}{*}{No.of} &
      \multicolumn{3}{c}{RT0} & 
      \multicolumn{3}{c}{VMS} \\     
   ID & nodes & elements & V-DOF &  P-DOF & Total-DOF & V-DOF &  P-DOF & Total-DOF \\ 
    \hline
    1 & 1379 & 2758	& 4061 & 2682 & 6743 & 2758	& 1379 & 4137\\
    
    2 & 2485 & 4970 & 7744 & 5128 & 12872 & 5232 & 2616 & 7848 \\
    
    3 & 4822 & 9644	& 14326 & 9504 & 23830 & 9644 &	4822 & 14466 \\
    
    4 & 10356 &	20712 & 30858 &	20502 &	51360 & 20712 &	10356 &	31068 \\
    
    5 &	17673 &	35346 &	52741 &	35068 &	87809 &	35346 &	17673 &	53019 \\
   
    6 &	26926 & 53852 &	80440 &	53514 &	133954 & 53852 & 26926 & 80778\\ 
   
    7 &	42911 & 85822 &	128301 & 85390 & 213691 & 85822 & 42911& 128733\\
    \hline
  \end{tabular}
  \label{table:mesh_id}
\end{table}
\begin{table}[t!]
\captionsetup{format=hang}
\centering
	\caption{Circular reservoir: $\beta_B $ values used for this study.}
	\begin{tabular}{c c c}
		\hline\hline
		Case ID & $\beta_B $  & Units \\	
		\hline		
		1 & $10^{-8}$ & $\mathrm{Pa^{-1}}$\\	
		2 & $10^{-7}$ & $\mathrm{Pa^{-1}}$\\		
		3 & $10^{-6}$ & $\mathrm{Pa^{-1}}$\\ 		
		\hline
	\end{tabular}
	\label{table:beta_case}
\end{table}  

Let $\mu = 10^{-3}$ Pa$\cdot$s and the anisotropic permeability tensor be given by:
 \begin{align}
  \mathbf{K} =
 \begin{pmatrix}
 \mathrm{cos(\theta)} &-\mathrm{sin(\theta)}\\
 \mathrm{sin(\theta)} &\mathrm{cos(\theta)}                    
 \end{pmatrix} 
 \begin{pmatrix}
 10^{-10}  &0 \\
 0    &10^{-13}  
 \end{pmatrix}
 \begin{pmatrix}
 \mathrm{cos(\theta)}  &\mathrm{sin(\theta)}\\
 -\mathrm{sin(\theta)} &\mathrm{cos(\theta)}
 \end{pmatrix}\;\mathrm{m}^2,
 \end{align}
 \begin{table}[t!]
 \captionsetup{format=hang}
 \centering
 \caption{Circular reservoir: percentage of DMP violations for RT0 and VMS formulations.}
   \begin{tabular}{ c| c c c |c c c}
     \hline
     \hline 
     \multirow{2}{*}{Mesh ID} &
     \multicolumn{3}{c|}{ \% Violations: RT0} &
     \multicolumn{3}{c}{\% Violations: VMS} \\

     & Case 1 & Case 2 & Case 3 & Case 1 & Case 2 & Case 3 \\
	 \hline
	        
	1 &	56.52 &	57.34 &	62.60 &	48.77 &	49.03 &	49.89	\\
	
	2 &	53.33 &	54.43 &	59.89 &	46.04 &	46.20 &	47.11	\\
	
	3 &	54.63 &	55.66 &	59.04 &	45.01 &	47.92 &	48.98 	\\
	
	4 &	53.19 &	53.92 &	56.73 &	47.95 &	50.60 & 51.55	\\
	
	5 &	51.06 &	51.60 &	54.46 &	51.50 &	46.29 &	50.13	\\
	
	6 &	52.85 &	51.87 &	53.01 &	47.25 &	48.81 &	52.58 	\\

	7 &	52.00 &	50.02 &	51.65 &	47.52 &	36.12 &	49.62\\

      \hline    
   \end{tabular}
   \label{table:Reser_violations}
 \end{table}
where $\theta = \frac{\pi}{3}$. Figures \ref{fig:homo_pres_rt_vio} 
and \ref{fig:homo_pres_vms_vio} show the 
pressures of RT0 and VMS formulations, respectively, for different values of $\beta_B$ 
before and after VI is imposed. It can be seen from these figures that DMP violations 
occur regardless of the $\beta_B$ used, and again the computational framework fixes 
these violations. Graphical representation of these violations are shown only for the 
first mesh (Mesh ID: 1), but detailed results concerning the violations for all 
other meshes are provided in Table \ref{table:Reser_violations}. 

 \begin{figure}[t!]
	\centering
	\subfigure[No VI ($\beta_B=10^{-8}$ Pa$^{-1}$)]{\includegraphics[scale=0.4]{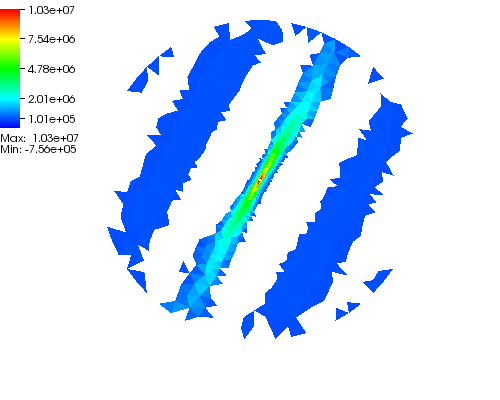}}
	\subfigure[With VI ($\beta_B=10^{-8}$ Pa$^{-1}$)]{\includegraphics[scale=0.4]{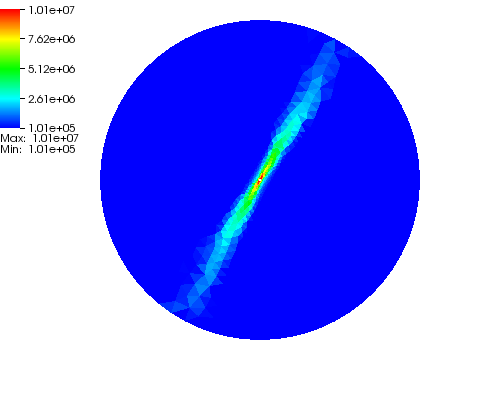}}\\
	\subfigure[No VI ($\beta_B=10^{-7}$ Pa$^{-1}$)]{\includegraphics[scale=0.4]{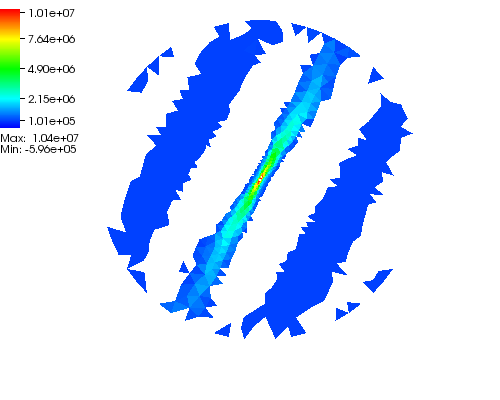}}
	\subfigure[With VI ($\beta_B=10^{-7}$ Pa$^{-1}$)]{\includegraphics[scale=0.4]{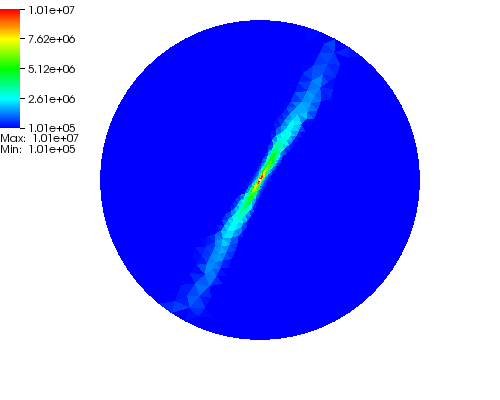}}\\
	\subfigure[No VI ($\beta_B=10^{-6}$ Pa$^{-1}$)]{\includegraphics[scale=0.4]{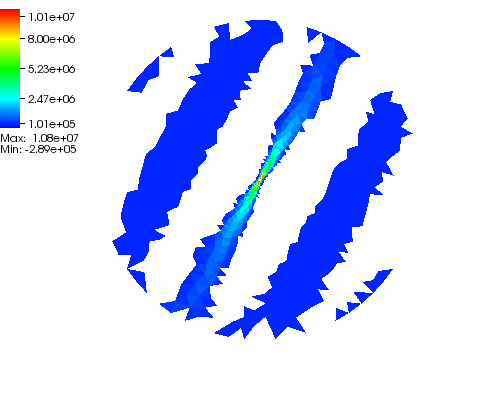}}
	\subfigure[With VI ($\beta_B=10^{-6}$ Pa$^{-1}$)]{\includegraphics[scale=0.4]{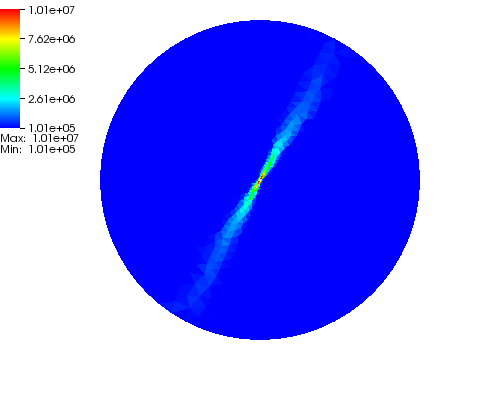}}		
\captionsetup{format=hang}
\caption{Circular reservoir: pressure contours for RT0 formulation before (left) 
and after (right) VI for various $\beta_B$. The white spaces are 
representative of DMP violations. It can be seen that $\beta_B$ has no 
affect on the initial violations.}
\label{fig:homo_pres_rt_vio}
\end{figure}
\begin{figure}[t!]
	\centering
	\subfigure[No VI ($\beta_B=10^{-8}$ Pa$^{-1}$)]{\includegraphics[scale=0.4]{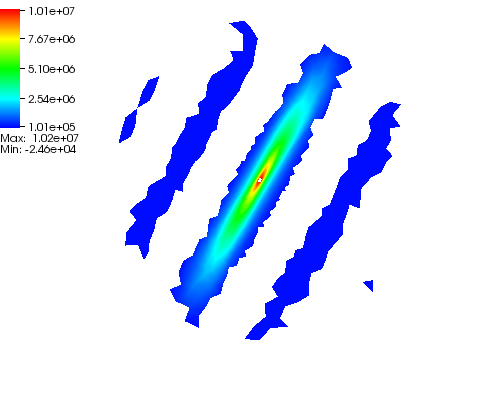}}
	\subfigure[With VI ($\beta_B=10^{-8}$ Pa$^{-1}$)]{\includegraphics[scale=0.4]{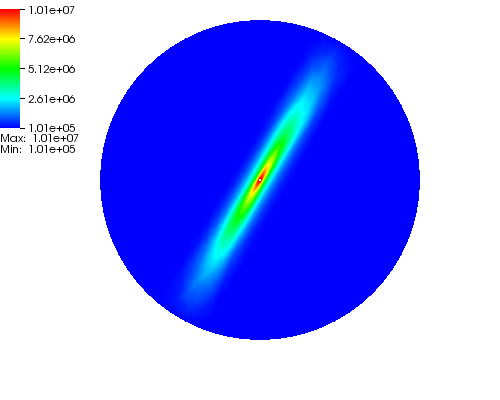}}\\
	\subfigure[No VI ($\beta_B=10^{-7}$ Pa$^{-1}$)]{\includegraphics[scale=0.4]{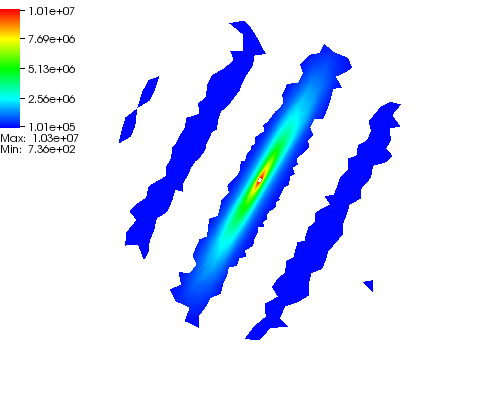}}
	\subfigure[With VI ($\beta_B=10^{-7}$ Pa$^{-1}$)]{\includegraphics[scale=0.4]{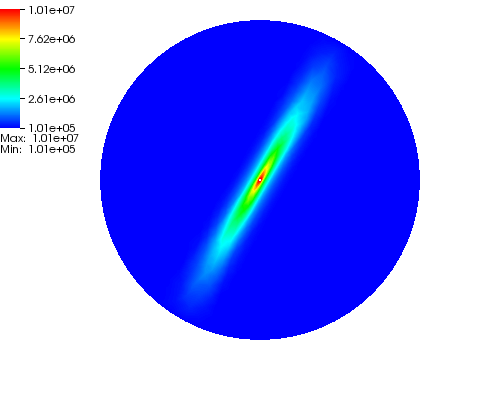}}\\
	\subfigure[No VI ($\beta_B=10^{-6}$ Pa$^{-1}$)]{\includegraphics[scale=0.4]{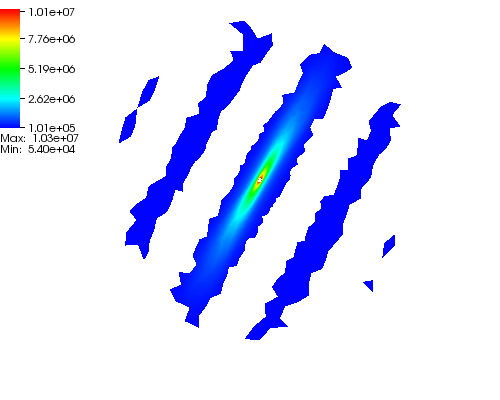}}
	\subfigure[With VI ($\beta_B=10^{-6}$ Pa$^{-1}$)]{\includegraphics[scale=0.4]{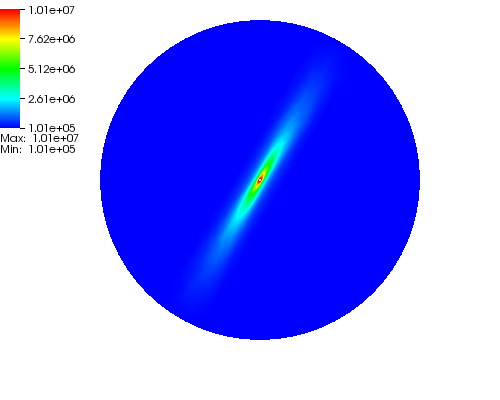}}		
\captionsetup{format=hang}
\caption{Circular reservoir: pressure contours for VMS formulation before (left) 
and after (right) VI for various $\beta_B$. The white spaces are 
representative of DMP violations. It can be seen that $\beta_B$ has no 
affect on the initial violations.}
\label{fig:homo_pres_vms_vio}
\end{figure}

\begin{table}[t]
\captionsetup{format=hang}
\centering
\caption{Circular reservoir: solver iterations and time-to-solution for RT0 formulation.}
  \begin{tabular}{ c |c c c| c c c | c }
    \hline
    \hline 
     \multirow{1}{*}{Mesh} &
       \multicolumn{3}{c|}{RT0} &
       \multicolumn{3}{c|}{VI over RT0} &
     \multirow{1}{*}{Total} \\
       
    ID & KSP & SNES & Time & KSP & SNES & Time & time  \\
    \hline 
	\hline
	 \multicolumn{8}{c}{$\beta_B=10^{-8}$ Pa$^{-1}$}\\
	
	    
	    \hline
	    \hline   
	1 & 375	& 3	& 7.95E-01	& 30 & 1 & 9.17E-02 & 8.86E-01 \\
	
	2 & 263	& 3	& 9.94E-01	& 31 & 1 & 1.13E-01 & 1.11E+00 \\
	
	3 & 303	& 3	& 2.00E+00	& 31 & 1 & 1.65E-01 & 2.17E+00 \\
	
	4 & 451	& 4	& 6.22E+00	& 34 & 1 & 2.83E-01 & 6.51E+00 \\
	
	5 & 708	& 5	& 1.63E+01	& 30 & 1 & 4.23E-01 & 1.67E+01 \\
	
	6 & 667	& 5	& 2.37E+01	& 33 & 1 & 6.68E-01 & 2.43E+01 \\
	
	7 & 868	& 6	& 4.85E+01	& 33 & 1 & 1.04E+00 & 4.95E+01 \\ 
     
     \hline
    \hline

     \multicolumn{8}{c}{$\beta_B=10^{-7}$ Pa$^{-1}$}\\
      \hline
      \hline   
         
	1 &	262	& 3	& 6.07E-01	& 30 & 1 & 8.82E-02 & 6.95E-01 \\

	2 &	267	& 3	& 9.97E-01	& 31 & 1 & 1.19E-01 & 1.12E+00 \\

	3 &	439	& 4	& 2.87E+00	& 31 & 1 & 1.69E-01 & 3.04E+00 \\

	4 &	462	& 4	& 6.48E+00	& 34 & 1 & 2.97E-01 & 6.77E+00 \\

	5 &	716	& 5	& 1.69E+01	& 31 & 1 & 4.52E-01 & 1.73E+01 \\

	6 &	681	& 5	& 2.46E+01	& 32 & 1 & 6.36E-01 & 2.52E+01 \\

	7 &	847	& 5 & 4.85E+01 & 33 & 1 & 1.03E+00 & 4.95E+01 \\
         
      \hline
      \hline

     \multicolumn{8}{c}{$\beta_B=10^{-6}$ Pa$^{-1}$}\\
      \hline
      \hline

	1 & 438	& 5	& 9.02E-01	& 30 & 1 & 8.82E-02 & 9.90E-01  \\

	2 &	431	& 5	& 1.65E+00	& 31 & 1 & 1.24E-01 & 1.77E+00 \\

	3 &	563	& 5	& 3.52E+00	& 30 & 1 & 1.56E-01 & 3.68E+00 \\

	4 &	546	& 5	& 7.47E+00	& 34 & 1 & 3.05E-01 & 7.77E+00  \\

	5 &	496	& 4	& 1.17E+01	& 31 & 1 & 4.43E-01 & 1.22E+01 \\

	6 &	507	& 4	& 1.84E+01	& 32 & 1 & 6.67E-01 & 1.90E+01 \\

	7 &	707	& 5 & 3.96E+01 & 34 & 1 & 1.01E+00 & 4.06E+01 \\

      \hline
      \hline      
  \end{tabular}
  \label{table:RT0_Iter}
\end{table} 
\begin{table}[t]
\captionsetup{format=hang}
\centering
\caption{Circular reservoir: solver iterations and time-to-solution for VMS formulation.}
  \begin{tabular}{ c |c c c| c c c | c }
    \hline
    \hline
    
     \multirow{1}{*}{Mesh} &
       \multicolumn{3}{c|}{VMS} &
       \multicolumn{3}{c|}{VI over VMS} &
     \multirow{1}{*}{Total} \\
       
     ID & KSP & SNES & Time & KSP & SNES & Time & time \\
    \hline
    \hline
    \multicolumn{8}{c}{$\beta_B=10^{-8}$ Pa$^{-1}$}\\
    
        \hline
        \hline   
	1 &	123	& 3	& 6.17E-01	& 37 & 2 & 1.82E-01 & 7.99E-01  \\  
	2 &	170	& 3	& 1.36E+00	& 48 & 2 & 3.10E-01 & 1.67E+00  \\

	3 &	191	& 3	& 2.88E+00	& 52 & 2 & 4.96E-01 & 3.38E+00 \\

	4 &	343	& 4	& 1.08E+01	& 72 & 2 & 1.02E+00 & 1.18E+01 \\

	5 &	417	& 4	& 2.51E+01	& 95 & 2 & 1.99E+00 & 2.71E+01 \\

	6 &	551	& 5	& 5.50E+01	& 114& 2 & 3.53E+00 & 5.85E+01 \\

	7 &	613	& 5	& 9.58E+01	& 133& 2 & 6.13E+00 & 1.02E+02 \\

     \hline
    \hline

      \multicolumn{8}{c}{$\beta_B=10^{-7}$ Pa$^{-1}$}\\
      \hline
      \hline   
         
	1 &	209	& 4	& 9.35E-01	& 35 & 2 & 1.78E-01 & 1.11E+00 \\
	
	2 &	234	& 4	& 1.80E+00	& 47 & 2 & 3.04E-01 & 2.10E+00 \\

	3 &	271	& 4	& 3.90E+00	& 50 & 2 & 4.96E-01 & 4.39E+00 \\
	
	4 &	386	& 4	& 1.22E+01	& 72 & 2 & 1.03E+00 & 1.32E+01 \\
	
	5 &	375	& 4	& 2.28E+01	& 95 & 2 & 2.03E+00 & 2.48E+01 \\
	
	6 &	567	& 5	& 5.35E+01	& 104& 2 & 3.25E+00 & 5.67E+01 \\
	
	7 &	574	& 5	& 9.03E+01	& 149& 2 & 6.91E+00 & 9.72E+01 \\

      \hline
      \hline

     \multicolumn{8}{c}{$\beta_B=10^{-6}$ Pa$^{-1}$} \\
      \hline
      \hline   
         
	1 &	209	& 5	& 9.69E-01	& 35 & 2 & 1.91E-01 & 1.16E+00 \\
	
	2 &	317	& 6	& 2.78E+00	& 46 & 2 & 3.20E-01 & 3.10E+00 \\
	
	3 &	368	& 5	& 5.23E+00	& 49 & 2 & 4.52E-01 & 5.68E+00 \\
	
	4 &	426	& 5	& 1.39E+01	& 65 & 2 & 9.58E-01 & 1.48E+01 \\
	
	5 &	554	& 5	& 3.46E+01	& 77 & 2 & 1.85E+00 & 3.65E+01 \\
	
	6 &	549	& 5	& 5.17E+01	& 98 & 2 & 3.07E+00 & 5.47E+01 \\
	
	7 &	602	& 5 &  9.49E+01  & 124 & 2 & 5.78E+00 & 1.01E+02 \\        
      \hline
  \end{tabular}
  \label{table:VMS_Iter}
\end{table}  
Solver performances are shown in Tables \ref{table:RT0_Iter} and \ref{table:VMS_Iter} for 
RT0 and VMS discretizations, respectively. First, it can be observed that as 
$\beta_B$ increases, the computational cost for both the VI and non-VI approaches 
also increases. Second, like the square reservoir problem, it can also be observed that
the increase in computational cost introduced by the VI solver is not significant; 
the total time-to-solution only increases by at most 15 percent for RT0 and 30 percent
for VMS. The number of SNES iterations required for the VI solvers remains 2 and 1 for 
RT0 and VMS, respectively, regardless of the mesh ID or $\beta_B$. However, the number 
of KSP iterations associated with VI over VMS increases with finer meshes.
\begin{figure}[t!]
		\centering
		\subfigure[RT0 and VMS only]{\includegraphics[scale=0.45]{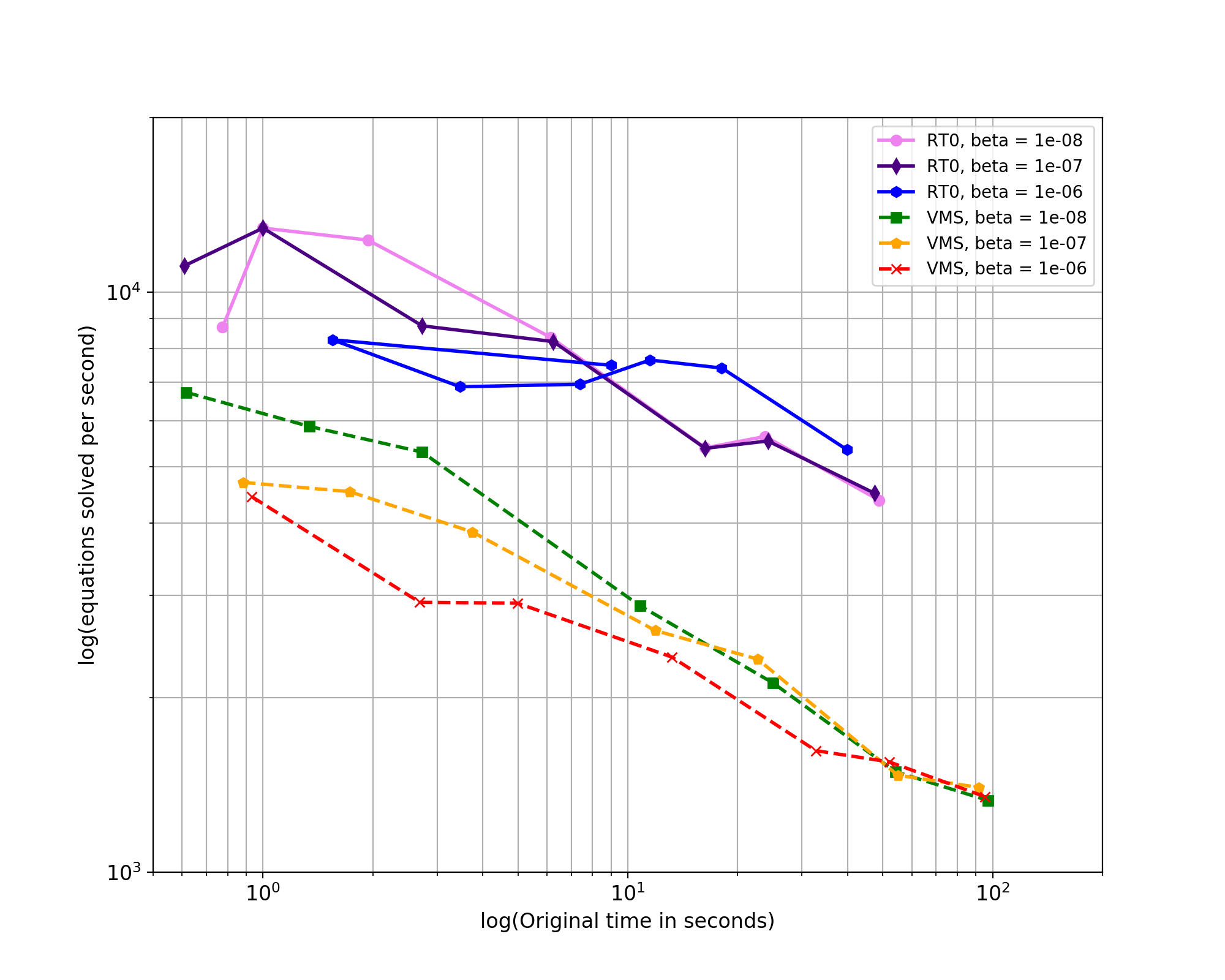}}\\
		\subfigure[VI only]{\includegraphics[scale=0.45]{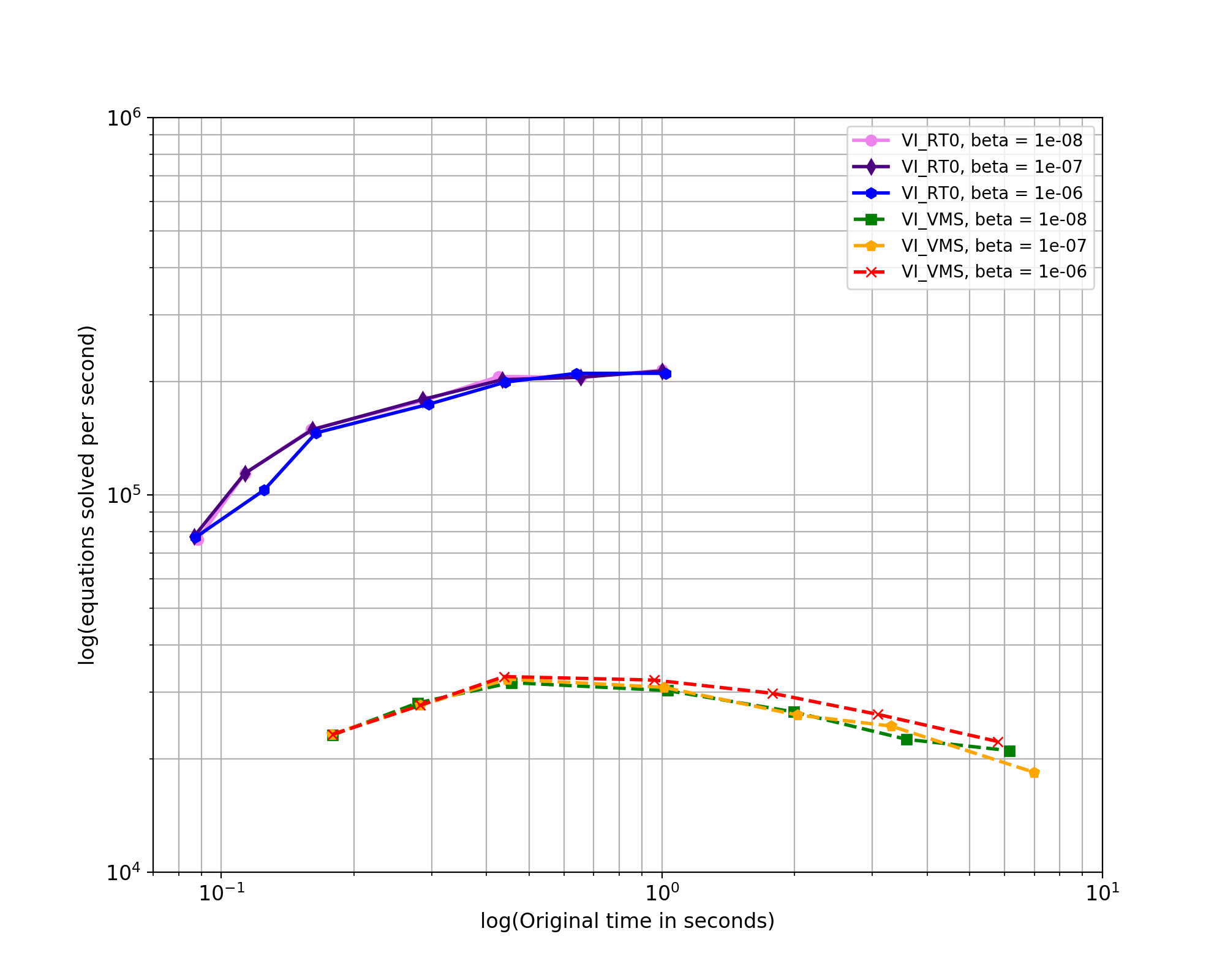}}
\captionsetup{format=hang}
	\caption{Static-scaling over VI and non-VI components of the framework 
        for different $\beta_B$ values and discretization. The flat horizontal lines 
        found in the VI only plots demonstrate
        excellent algorithmic scalability in comparison to the RT0 and VMS only
        components.}
	\label{fig:static_scaling_components}
\end{figure}
\begin{figure}[t!]
		\centering
		\subfigure[Rate-metric vs Total time taken]{\includegraphics[scale=0.45]{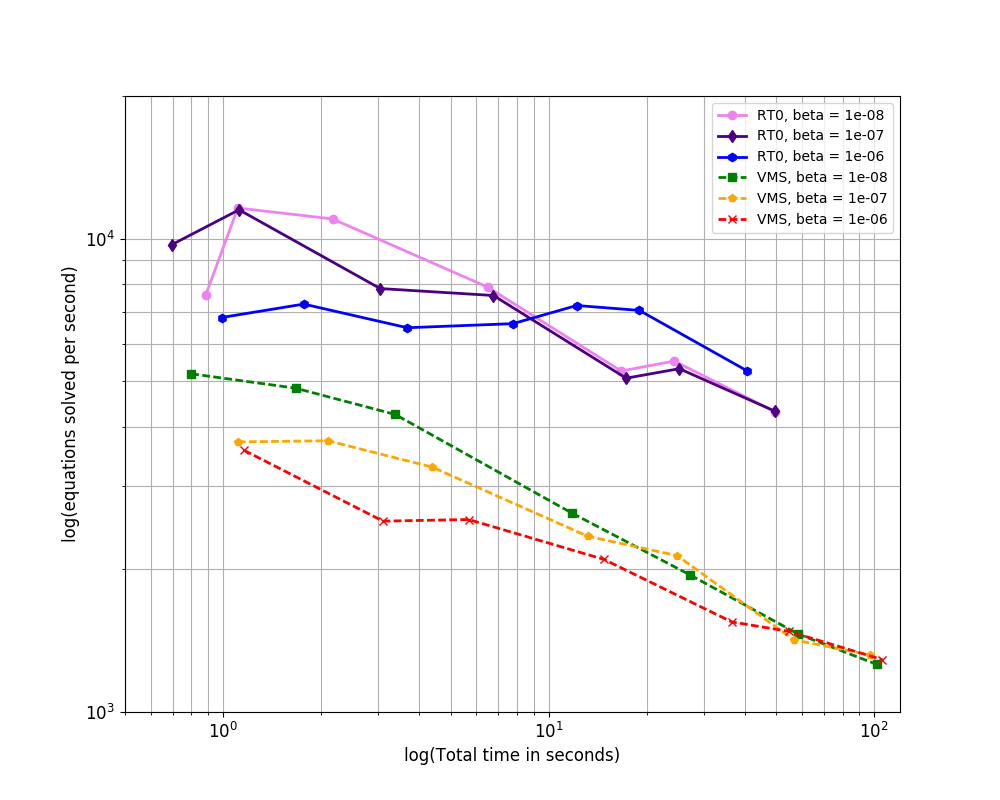}}\\
		\subfigure[Rate-metric vs Total no.of degrees of freedom solved]{\includegraphics[scale=0.45]{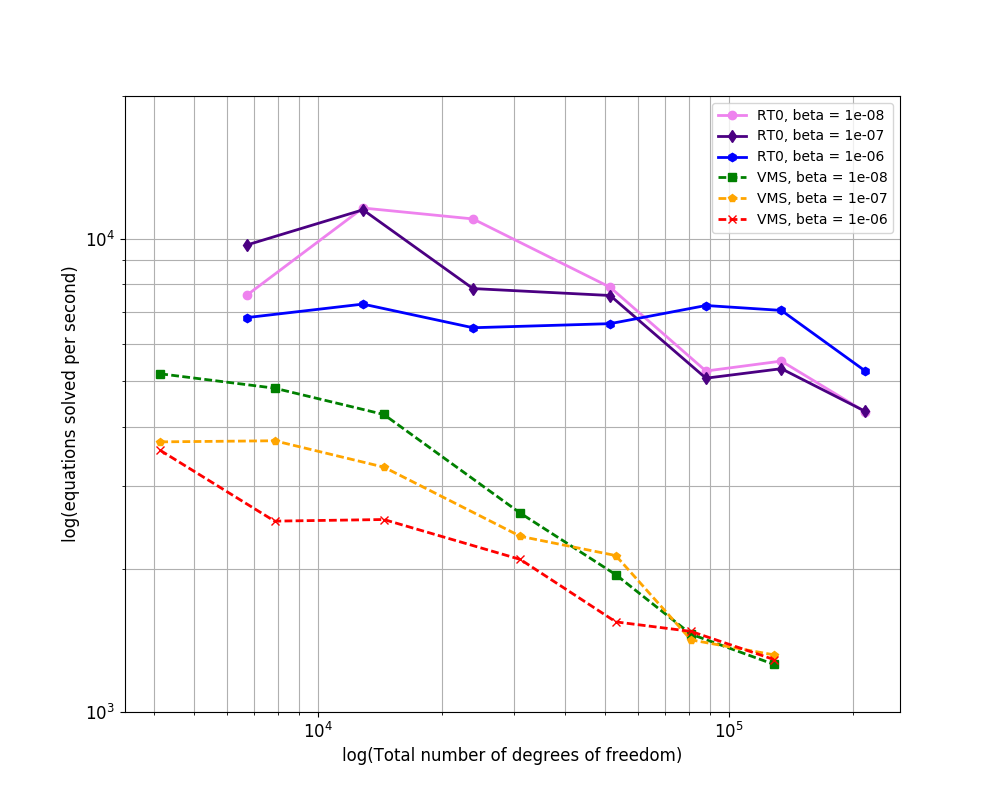}}
\captionsetup{format=hang}
	\caption{Static-scaling over total time for different $\beta_B$ values. The RT0
        formulation not only has better algorithmic convergence but is capable of
        solving more degrees-of-freedom per second.}
	\label{fig:static_scaling_overall}
\end{figure}
In order to understand algorithmic scalability, a rate metric of some sort is 
needed to understand the performance of the SNES and KSP solvers. Herein, 
\emph{static-scaling} plots as described in \citep{chang2017performance} shall be used
for this purpose. Static-scaling plots document the degrees-of-freedom solved per second
across all mesh IDs. The degrees-of-freedom solved per second for each phase of
the computational frameworks are shown in Figure \ref{fig:static_scaling_components}. 
Although the initial guess components of the RT0 and VMS solvers demonstrate suboptimal
convergence, their respective VI components have excellent algorithmic scaling. 
Figure \ref{fig:static_scaling_overall} combines the static-scaling plots of the overall computational
effort for both mixed formulations. It is well-known that RT0 has more degrees of freedom (DOF) 
than VMS discretization for any given mesh, but it can be seen that the rate metrics are 
higher for RT0 when compared with VMS. This suggests that for the same problem size, the 
RT0 formulation is in fact more efficient for the proposed VI framework. 
Furthermore, the rate metric for the VMS formulation decreases significantly as the 
problem size increases. These findings are consistent with the fact that the KSP 
iteration count increases for VI over VMS as the mesh is refined. Even though 
RT0 has a slower theoretical convergence rate than VMS as seen from Figure \ref{fig:h_conv},
it is in fact more computationally efficient because it solves more 
degrees-of-freedom per second and is also more algorithmically 
efficient because the tailing off towards the right is not as significant..
\FloatBarrier
\subsection{3D reservoir problem}

\begin{figure}[hbt]
  \centering
  \captionsetup{format=hang}
  \subfigure{\includegraphics[scale =0.7]{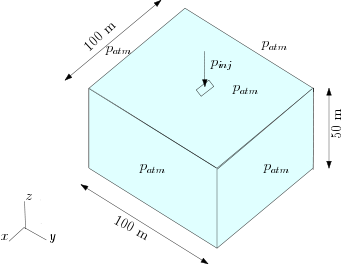}}
  \subfigure{\includegraphics[scale=0.4]{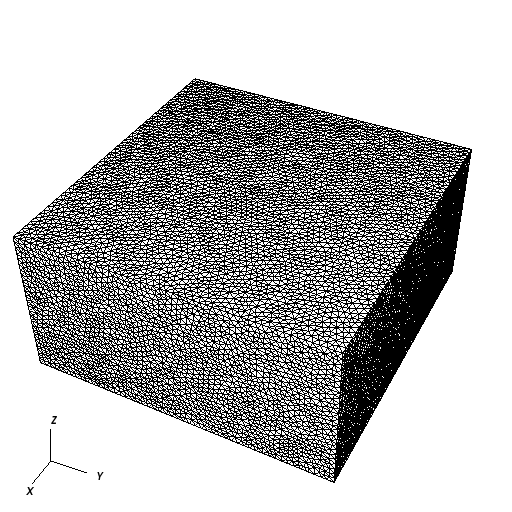}}
  \captionsetup{format=hang}
  \caption{3D reservoir problem:~The left figure provides
    a pictorial description of the problem, and the right
    figure shows the corresponding unstructured finite
    element mesh using tetrahedral elements.
    \label{fig:3D_res_q}}
\end{figure}
For this last problem, we consider a 3D reservoir which 
is larger in size than the previous problems and study 
the parallel scalability of the VI framework. Consider 
a three dimensional cuboid with dimensions $100 \times 
100 \times 50 \; \mathrm{m}$. The outer surfaces are 
maintained at atmospheric pressure ($p_{atm} = 1$ atm). 
An injection pressure ($p_{inj}$) of 
\begin{equation}
  p_{inj} =  1 + 10\times 
  \sin\left(\pi\frac{(x-48)}{4}\right)
  \sin\left(\pi\frac{(y-48)}{4}\right) \;\mathrm{atm},
\end{equation}
is applied over the square region $ [48,52] 
\times [48,52]$ on the top surface, see Figure 
\ref{fig:3D_res_q} for a pictorial description 
of this problem. We take $\rho\mathbf{b} = 
\mathbf{0}$, $f=0$, $\mu = 10^{-3}$ Pa$\cdot$s, 
$\beta_B = 10^{-8}$ Pa$^{-1}$ and the permeability 
tensor to be as follows: 
\begin{align}
\mathbf{K} = \begin{pmatrix}
10^{-13} & 0 & 0\\
0 & 10^{-13} & 0 \\
0 & 0 & 10^{-11}
\end{pmatrix}\;\mathrm{m}^2.
\end{align}
The 3D domain is discretized using only the RT0 formulation 
resulting in 549,023 velocity degrees-of-freedom, 267,869 
pressure degrees-of-freedom, and hence 816,892 total number 
of degrees-of-freedom. Strong-scaling is conducted up to 64 
cores on a KNL processor.

\begin{figure}[t]
  \centering
  \subfigure[Before VI]{\includegraphics[scale=0.4]{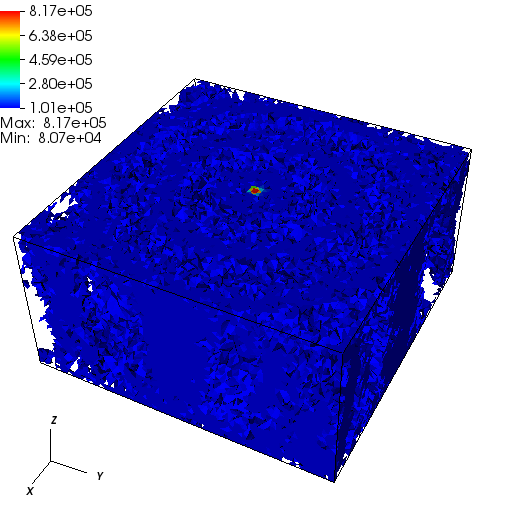}}
  \subfigure[After VI]{\includegraphics[scale=0.4]{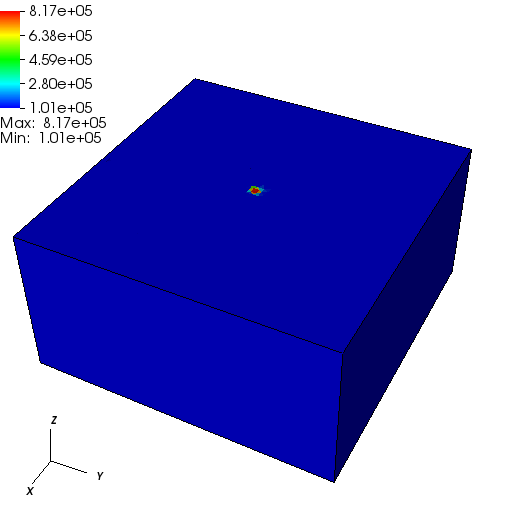}}
  \captionsetup{format=hang}
  \caption{3D reservoir problem:~Pressure profiles of RT0 and RT0 + VI formulations. The clipped regions denote DMP violations.}
  \label{fig:3D_mesh_pres}
\end{figure}

\begin{table}[t]
\captionsetup{format=hang}
\centering
\caption{3D reservoir problem:~Strong-scaling with respect to one core.}
  \begin{tabular}{ c |c c c| c c c }
    \hline
    \hline 
     \multirow{1}{*}{No.of} &
       \multicolumn{3}{c}{Time} &
       \multicolumn{3}{c}{Parallel efficiency ($eff_n$)} \\       
    Cores & RT0  & RT0 + VI & Total  & RT0  & RT0 + VI & Total  \\
    \hline 
	\hline

	1 &	4.09E+002 &	2.73E+002 &	6.82E+002 &	100.00 & 100.00	& 100.00	\\
	2 &	1.97E+002 & 1.78E+002 &	3.75E+002 &	103.83 & 76.44 & 90.82\\
	4 &	1.08E+002 &	9.63E+001 &	2.04E+002 &	94.48 &	70.78 &	83.32 \\
	8 &	6.39E+001 &	5.67E+001 &	1.21E+002 &	79.98 &	60.07 &	70.61 \\
	16 & 3.80E+001 & 3.24E+001 & 7.04E+001 & 67.32 & 52.54 & 60.51 \\
	32 & 2.42E+001 & 2.39E+001 & 4.81E+001 & 52.74 & 35.63 & 44.24 \\
	64 & 1.65E+001 & 1.42E+001 & 3.07E+001 & 38.70 & 29.93 & 34.64	\\
      \hline     
  \end{tabular}
  \label{table:core_vs_time}
\end{table}  
\begin{table}[t]
\captionsetup{format=hang}
\centering
\caption{3D reservoir problem:~Strong-scaling with respect to two cores.}
  \begin{tabular}{ c |c c c| c c c }
    \hline
    \hline 
     \multirow{1}{*}{No.of} &
       \multicolumn{3}{c}{Time} &
       \multicolumn{3}{c}{Parallel efficiency ($\%$)} \\       
    Cores & RT0  & RT0 + VI & Total  & RT0  & RT0 + VI & Total  \\
    \hline 
	\hline

	2 &	1.97E+002 & 1.78E+002 &	3.75E+002 &	100.00 & 100.00 & 100.00\\
	4 &	1.08E+002 &	9.63E+001 &	2.04E+002 &	91.20 &	92.42 &	91.91 \\
	8 &	6.39E+001 &	5.67E+001 &	1.21E+002 &	77.07 &	78.48 &	77.48 \\
	16 & 3.80E+001 & 3.24E+001 & 7.04E+001 & 64.80 & 68.54 & 66.58 \\
	32 & 2.42E+001 & 2.39E+001 & 4.81E+001 & 50.87 & 46.55 & 48.73 \\
	64 & 1.65E+001 & 1.42E+001 & 3.07E+001 & 37.31 & 39.17 & 38.17	\\
      \hline     
  \end{tabular}
  \label{table:core_vs_time2}
\end{table} 
\begin{table}[t!]
 \captionsetup{format=hang}
 \centering
 \caption{3D reservoir problem:~Change in number of iterations for different number of cores.}
   \begin{tabular}{ c|c c |c c }
     \hline
     \hline 
     \multirow{1}{*}{No.of} &
     \multicolumn{2}{c}{ RT0} &
     \multicolumn{2}{c}{ RT0 + VI} \\

     Cores & KSP  & SNES & KSP  & SNES \\
	 \hline
	        
	1	&	237	&	2	&	219	&	3	\\
	2	&	211	&	2	&	279	&	3	\\
	4	&	226	&	2	&	289	&	3	\\
	8	&	240	&	2	&	301	&	3	\\
	16	&	251	&	2	&	299	&	3	\\
	32	&	284	&	2	&	404	&	3	\\
	64	&	277	&	2	&	411	&	3	\\

      \hline    
   \end{tabular}
   \label{table:core_vs_iter}
 \end{table}

Figure \ref{fig:3D_mesh_pres} shows the RT0 pressures before and after VI is imposed. 
The computational framework successfully eliminates all DMP violations even for a larger
3D problem. In Table \ref{table:core_vs_time}, it can be seen that the parallel 
scalability of the RT0 + VI combination is somewhat similar to the RT0 only framework.
A noticeable deterioration in parallel efficiency is noticed when the VI framework jumps
from 1 core  (i.e., serial) to 2 cores (i.e., parallel). However, if strong-scaling 
was conducted from 2 cores and on, the parallel scalability is nearly identical as seen
from Table \ref{table:core_vs_time2}. The added computational cost associated with 
enforcing the bound constraints is slightly larger than observed from the last 
two 2D reservoir problems because the total time is now nearly doubled. Nevertheless,
the SS method utilized in this paper is comparatively less expensive than the RS 
framework used in \citep{chang2017variational}, and it is possible that different 
numerical discretizations may tell a different story. Lastly, Table 
\ref{table:core_vs_iter} depicts the number of KSP and SNES iterations 
required for different core counts, and the numbers
remain relatively consistent. Thus, we can conclude based off all these 
computational results of this problem that the VI approach for the modified Darcy
model with pressure-dependent viscosity has 
comparable parallel scalability for particular mixed formulations like the 
RT0 formulation. 
\FloatBarrier

%% file: Sections/S6_NN_CR.tex
\section{CONCLUDING REMARKS}
\label{Sec:S6_NN_CR} 
The VI-based formulation proposed in this
paper is a comprehensive framework that
enforces maximum principles for flow
through porous media models which
account for pressure-dependent
viscosity and anisotropy.
%
Some of the \emph{salient features} of the
proposed formulation are as follows:
\begin{enumerate}[(S1)]
\item To the best of our knowledge, this is the only
  computational framework that can enforce DMP even
  for anisotropic and nonlinear flow through porous
  media models.
\item The proposed VI based framework works on
  any mixed finite element weak formulation, as
  demonstrated through using the RT0 and VMS
  formulations. The underlying weak form can
  be non-symmetric and non-linear. 
\item The formulation allows the user to place
  desired bounds on the field variables like
  the maximum principle on the pressure field.
\item The formulation is amenable for an
  implementation in a parallel environment. 
\end{enumerate}

The \emph{main findings} of our study
are summarized as follows: 
\begin{enumerate}[(C1)]
\item The convergence study (reported in
  subsection \ref{sec:h_conv}) for the Firedrake implementation
  of the modified Darcy equation with pressure-dependent viscosity 
  indicates that the computational framework has an $L_2$ error
  convergence rate that matches the theoretical convergence rate.
\item  It is shown that the extent of anisotropy and heterogeneity has a 
  direct impact on the percentage of DMP violations. These violations 
  tend to decrease when anisotropy is decreased.
\item The study also infers that mesh refinement
  does not reduce the percentage of DMP violations
\item Our study also shows that the degree of
  nonlinearity (i.e., extent of viscosity
  dependence on pressure), though impacting
  solver performance, has no significant
  influence on the percentage of DMP violations.
\item The number of KSP and SNES iterations do not vary much when either problem size
or number of cores increases suggesting that the VI framework is algorithmically scalable.
\item The static-scaling plots (presented in 
  subsection \ref{sec:s_s}) reveals that the
  the VI component of the computational framework is
  much more scalable in the algorithmic sense than the 
  standard Newton solvers used for computing the initial 
  RT0 and VMS guesses. 
\item Furthermore, the degrees-of-freedom solved per second
  for certain formulations (e.g., VMS) decreases 
  as the problem size increases. It also sheds light on the 
  fact that for a given mesh, the RT0 discretization has better 
  static-scaling than that of the VMS discretization, reinstating
  that this static-scaling study can be used as a
  reference guide to compare not only numerical accuracy but
  also computational costs of various discretization.
\item We have shown that the parallel performance for
  the VI framework is comparable to the standard Newton
  approach for solving standard nonlinear equations.   
\end{enumerate}

A possible future work can be towards developing
a VI-based framework for multi-phase flows through
porous media that respects maximum principles on
general computational grids.